\documentclass[twoside,final,11pt]{entics} 
\usepackage{enticsmacro}
\usepackage{graphicx}
\usepackage[all]{xy}
\usepackage{bussproofs}

\usepackage{amsmath,amssymb}
\usepackage{tikz-cd}
\usepackage{hyperref}
\usepackage{quiver}
\usepackage{enumitem}
\usepackage{array}

\usepackage{todonotes}

\def\conf{MFPS 2026} 	
\volume{6}			
\def\volu{6}			
\def\papno{12}			
\def\mydoi{19649}
\def\lastname{Mio, Riba} 
\def\runtitle{Compact Quantitative Equational Theories as Monad Restrictions} 
\def\copynames{M. Mio, C. Riba}    

\newcommand\sle\subseteq
\newcommand\gle\supseteq

\newcommand\comp{\mathbin{\circ}}
\newcommand{\into}{\hookrightarrow} 

\newcommand\Str{\operatorname{Str}}
\newcommand\Mod{\operatorname{Mod}}
\newcommand\EM{\operatorname{EM}}
\newcommand\id{\operatorname{id}}	
\newcommand\Id{\operatorname{Id}}	

\newcommand\Pres{\operatorname{Pres}}

\newcommand\cat[1]{\mathcal{#1}}

\newcommand\Set{\mathbf{Set}}
\newcommand\Met{\mathbf{Met_1}}
\newcommand\FRel{\mathbf{FRel}}

\newcommand\Sig{\mathcal{S}}		
\newcommand\Sigbis{{\Sig'}}	

\newcommand\Horn{\mathcal{H}}		
\newcommand\Hornbis{{\Horn'}}	

\newcommand\atom{\alpha}		
\newcommand\atombis{{\atom'}}		
\newcommand\atomext{\atom'}

\begin{document}
\begin{frontmatter}
    \title{Compact Quantitative Equational Theories\\ as Monad Restrictions}
         \thanks[ALL]{This work has been partially supported by the french ANR PRCI project ATQUA (ANR-25-CE48-1428-01).}  
  \author{Matteo Mio\thanksref{a}\thanksref{myemail}}	
   \author{Colin Riba\thanksref{b}\thanksref{coemail}}		
   \address[a]{CNRS and ENS de Lyon, France}  							
   \thanks[myemail]{Email: \href{mailto:matteo.mio@ens-lyon.fr} {\texttt{\normalshape
        matteo.mio@ens-lyon.fr}}} 
  \address[b]{ENS de Lyon, France}  
  \thanks[coemail]{Email:  \href{mailto:colin.riba@ens-lyon.fr} {\texttt{\normalshape
        colin.riba@ens-lyon.fr}}}
\begin{abstract} 
Proofs in quantitative algebra are, generally, infinite objects due to the presence of infinitary rules in the deductive apparatus. 
A quantitative equational theory is called compact if all its consequences are derivable by means of finite proofs, that is, proofs not using any instances of the infinitary rules.
This definition is purely logical.
In this work we give an equivalent, but categorical, characterisation based on the notion of monad restriction. 
\end{abstract}
\end{frontmatter}

\section{Introduction}

\noindent
Quantitative algebra, introduced in \cite{DBLP:conf/lics/MardarePP16},
is an active area of investigation (see, e.g., \cite{DBLP:journals/jacm/BacciMPP26,adamek22,adv23,DBLP:journals/lmcs/BacciBLM18,bmpp18,DBLP:journals/corr/abs-2302-01224,DBLP:journals/lmcs/BacciMPP24,Gavazzo18}, 
\cite{DBLP:journals/pacmpl/GavazzoF23,DBLP:conf/lics/MardarePP17,Sarkis2024,MioEA24,MioEA22,DBLP:conf/lics/MioSV21,DBLP:conf/concur/MioV20,DBLP:conf/fscd/LagoHLP22,JurkaEA24,MiliusUrbat19,WildS20,WildS21,GoncharovEA23,DAngeloEA24,ForsterEA25,entics:16876})
aimed at extending the methods of universal algebra to the study of so-called
\emph{quantitative algebras}.
In the original formulation of \cite{DBLP:conf/lics/MardarePP16}, a quantitative algebra 
is an algebra whose carrier is an extended metric space.
More formally, it is a triple
$\mathbb{A}=(A,d_A,\{ \sigma^{A} \}_{\sigma\in\Sigma})$
where $(A,d_A)$ is an extended metric space and
$(A,\{ \sigma^{A} \}_{\sigma \in\Sigma})$ is a $\Sigma$-algebra,
in the ordinary sense of universal algebra, for
some signature $\Sigma$ of function symbols.
The key novel concept is that of \emph{quantitative equation}:
an expression of the form $s=_\epsilon t$ indexed by a real number $\epsilon\geq 0$,
whence the adjective ``quantitative'', which generalises ordinary equations ($s=t$) 
and makes it possible to axiomatize
the interplay between the algebraic and the metric structure.
The meaning of  $s=_\epsilon t$ is that the distance between any interpretations
of the terms $s$ and $t$ is bounded above by $\epsilon$,
i.e., $d_A(s^A, t^A)\leq\epsilon$. 

\vspace{0.2cm}
\emph{1.1 Monads. }
Given a collection of quantitative equations $T$ and an extended metric space $(X,d_X)$, it is possible to construct the free quantitative algebra satisfying $T$ generated by $(X,d_X)$.
This construction gives rise to a monad on the category $\mathbf{Met}$ of extended metric spaces and nonexpansive maps. Monads on $\mathbf{Met}$ that can 
be axiomatized by quantitative equations include the Kantorovich probability distribution monad \cite{DBLP:conf/lics/MardarePP17}, the Hausdorff powerset monad \cite{DBLP:conf/lics/MardarePP17} and the Kantorovich--Hausdorff convex powerset of probability distributions monad \cite{DBLP:conf/concur/MioV20}.  
These $\mathbf{Met}$ monads have applications in 
semantics and formal verification for modelling computational
effects in programs having quantitative behaviours
(such as, e.g., probabilistic programs).
The attractiveness of quantitative algebra in this context
is to provide a logical/syntactic way to handle such monads, in the same way that 
ordinary universal algebra and its logical/syntactic apparatus of
equational logic are used to handle (finitary) $\Set$ monads
(see, e.g., \cite[3.18]{Adamek_Rosicky_1994}).

\vspace{0.2cm}
\emph{1.2 Compact Quantitative Equational Theories. }
A main obstacle towards the above outlined application 
is that the logical apparatus of quantitative algebra
is based on \emph{infinitary first order logic}. 
Concretely, this means that the binary relations appearing in quantitative equations  ($=_\epsilon$, for $\epsilon \geq 0$) are subject to, among others, the following axiom:
\[
\forall x,y. \ \ \Big( \displaystyle \bigwedge_{n\in\mathbb{N}} x =_{\epsilon_n} y  \Big) \Rightarrow x =_\epsilon y \qquad \qquad \textnormal{ where } \epsilon=  \inf \{ \epsilon_n \}_{n\in \mathbb{N}}
\tag{Infinitary Axiom}
\]

\noindent
It states that if the distance between $x$ and $y$ is bounded above by a set
$\{ \epsilon_n \}_{n\in\mathbb{N}}$ of upper bounds, 
then necessarily the distance is also bounded by the infimum of all these upper bounds.  
The presence of this infinitary axiom means that, in general,
proofs are infinitely branching well-founded trees.
Hence they are not finite objects 
and therefore cannot be represented, manipulated, enumerated and verified by machines.

On the one hand, it can be shown that an infinitary axiom, such as the one above, is necessary in any sound and complete proof system for quantitative algebra. 
On the other hand, it has been recently observed in \cite{entics:16876} that, for some well behaved theories called \emph{compact quantitative equational theories} $T$, it holds that a quantitative 
equation is derivable from $T$ in the proof system if and only if it is derivable  from $T$ \emph{without} any usage of infinitary axioms. 
Hence, compact quantitative equational theories are exactly those amenable to finitary syntactic manipulation and proofs, just as ordinary equational theories of universal algebra. Perhaps surprisingly, \cite{entics:16876} has shown that rather than being an artificial or rare notion, the class 
 of compact quantitative equational theories includes many interesting examples such as those axiomatizing the $\mathbf{Met}$ monads, mentioned earlier, useful in program semantics and verification.

\vspace{0.2cm}
\emph{1.3 Technical Contribution. }
The definition of \emph{compact} quantitative equational theory
from~\cite[Def.~3.1]{entics:16876},
as briefly outlined above, is purely syntactical: a theory $T$ is compact
whenever there are finite proofs for all derivable quantitative equations.
The main achievement of the present work is to give an equivalent characterisation
of compactness using the language of category theory and monads. 

Rather than working with the notion of quantitative algebra
of~\cite{entics:16876} or~\cite{DBLP:conf/lics/MardarePP16},
it is technically convenient to work with an even more
general definition inspired from both~\cite{FordEA21} and~\cite{MioEA24}.
Namely, the base category $\mathbf{Met}$ is replaced and generalised to
$\Mod_{\Pi}(\mathbb{B})$,
the category of models of $\mathbb{B}$ and their homomorphisms, where 
 $\Pi$ is a relational signature consisting only of relation symbols and
$\mathbb{B}$ is a collection of (possibly infinitary) Horn formulas over $\Pi$.
The special case $\mathbf{Met}$ is obtained by taking
$\Pi = \{ {=_\epsilon} \mid \epsilon \geq 0 \}$
and by defining the collection $\mathbb{B}$ 
to include, among others,
the infinitary axiom mentioned above
(see Example~\ref{example:base-theory-fuzzy-relations} and~\cite[\S 2.3]{MioEA24}
for a full list).
The algebraic part, on top of the base category, 
is specified by a signature $\Sigma$ consisting only of function symbols (this includes constants, as nullary functions) and a set
$\mathbb{Q}$ of (possibly infinitary) Horn formulas over the extended signature
$\Pi+ \Sigma$ having a special shape
and playing the role of quantitative equations
(see Section~\ref{sec:quantitative-algebra-background} for details).
The category $\Mod_{\Pi+\Sigma}(\mathbb{B}+\mathbb{Q})$,
over the extended signature and theory,
plays the role of the category of quantitative algebras.
The free algebra construction then gives rise to the following adjunction
$F\dashv U$ and corresponding monad $ T= U F$ on $\Mod_\Pi(\mathbb{B})$:
\begin{center}
   \hspace{1.5cm} Base Category \qquad Category of Quantitative Algebras
\end{center}
\[\begin{tikzcd}
	{\Mod_\Pi(\mathbb{B})} & {\Mod_{\Pi+\Sigma}(\mathbb{B}+\mathbb{Q})}
	\arrow["F", shift left=2, from=1-1, to=1-2]
	\arrow["U", shift left, from=1-2, to=1-1]
\end{tikzcd}.\]
If we take the theories $\mathbb{B}$ and $\mathbb{Q}$ and remove from them all infinitary Horn formulas, thus only leaving the finite formulas, we obtain finitary sub-theories $\mathbb{B}_f\subseteq \mathbb{B}$ and $\mathbb{Q}_f\subseteq \mathbb{Q}$ and the following adjunction $F_f\dashv U_f$ and monad $T_f = U_f F_f$ on $\Mod_{\Pi+\Sigma}(\mathbb{B}_f):$
\begin{center}
	$
	\begin{tikzcd}
	{\Mod_\Pi(\mathbb{B}_f)} & {\Mod_{\Pi+\Sigma}(\mathbb{B}_f+\mathbb{Q}_f)}
	\arrow["{F_f}", shift left=2, from=1-1, to=1-2]
	\arrow["{U_f}", shift left, from=1-2, to=1-1]
	\end{tikzcd}.
	$
\end{center}

\noindent
We denote with
$G\colon \Mod_{\Pi}(\mathbb{B})\hookrightarrow  \Mod_{\Pi}(\mathbb{B}_f)$
and
\(
  I\colon
  \Mod_{\Pi+\Sigma}(\mathbb{B}+\mathbb{Q})
  \hookrightarrow
  \Mod_{\Pi+\Sigma}(\mathbb{B}_f+\mathbb{Q}_f)
\)
the obvious inclusions
(every model of $\mathbb{B}$ (resp.\ $\mathbb{B} + \mathbb{Q}$)
is also a model of $\mathbb{B}_f$
(resp.\ $\mathbb{B}_f + \mathbb{Q}_f$)).


Our main result is Theorem~\ref{thm:main-result} in Section~\ref{sec:main-result}.
It states a correspondence between a purely logical property of the theory $T$
and a categorical property of the corresponding monad: 

\textbf{Main Result:}
The quantitative equational theory given by $\mathbb{B}$ and $\mathbb{Q}$ is compact
(Definition~\ref{compactness_definition})

\hspace{2.57cm} if, and only if,
$T$ is a monad restriction\footnote{
	A monad restriction is a natural isomorphism $\gamma: T_f G\to G T$ subject to compatibility constraints (\S\ref{sec:monad-background}). 
} \hspace{-0.1cm}of $T_f$
and the restriction is compatible with $I$.

This paper is organised as follows.
In Section~\ref{sec:tech-background} the required notions from category theory, including monad restrictions, and  of Horn logic are given. 
In Section~\ref{sec:quantitative-algebra-background} we define, following \cite{FordEA21}, the generalised notion of quantitative algebra outlined above. 
In Section~\ref{sec:compact-definition}, the definition of compact quantitative equational theory of \cite{entics:16876} is adapted to this more general setting.
The definition coincides with that of \cite{entics:16876} when the quantitative algebras of \cite{entics:16876} are taken as particular case.
In Section~\ref{sec:main-result} we present and give a full and self contained proof of our main Theorem~\ref{thm:main-result}. 
Lastly, we discuss directions for future work in Section~\ref{sec:conclusions}.

\section{Technical Background}\label{sec:tech-background}
Throughout this background section, we assume basic knowledge of category theory and logic. 

\subsection{Eilenberg--Moore category and Canonical Comparison Functor} \label{sec:monad-background-comparison}
Given categories $\mathcal{C}$ and $\mathcal{D}$ and an adjunct pair
$F\dashv U \colon \mathcal{C} \to \mathcal{D}$ with unit $\eta$ and counit $\epsilon$,
\[
  F\colon \mathcal{C} \to \mathcal{D}
  \qquad
  U\colon \mathcal{D} \to \mathcal{C}
  \qquad
  \eta\colon \Id_{\mathcal{C}}\Rightarrow U F
  \qquad
  \epsilon\colon F U \Rightarrow  \Id_{\mathcal{D}}
\]
there is a corresponding monad $M=(M,\eta^M, \mu^M)$ on $\mathcal{C}$,
with underling functor $M:= U F$,
unit $\eta^M:=\eta$ and multiplication $\mu^M:=(U\epsilon_F)$ (see, e.g.,~\cite[VI.1]{maclane98book}).

Different adjunctions may induce the same monad $M$ on $\mathcal{C}$.
A canonical choice is given by the Eilenberg--Moore construction
(see, e.g.,~\cite[VI.2]{maclane98book}).
Given a monad $M=(M,\eta^M,\mu^M)$ on $\mathcal{C}$ we denote with $\EM(M)$ the Eilenberg--Moore category of $M$.
The objects of $\EM(M)$ are called $M$-algebras and are pairs
$(C,\alpha_C)$ with $C\in \mathcal{C}$ and
$\alpha_C\colon MC \to C$ in $\mathcal{C}$ satisfying:
$\alpha_C \circ \eta_C = \id_C$
and $ \alpha_C \circ M(\alpha_C) = \alpha_C \circ \mu_C$. 
The morphisms of $\EM(M)$ from $(C,\alpha_C)$ to $(C',\alpha_{C'})$ are
morphisms $f\colon C\to C'$ in $\mathcal{C}$ satisfying 
\(
f\circ \alpha_C  = \alpha_{C'}\circ M(f)
\). 
There is a forgetful functor $U_{M}\colon \EM(M)\to \mathcal{C}$ defined on objects as
$U_{M}(C,\alpha_C)=C$ and on morphisms as $U_{M}(f)=f$.
This functor has a left adjoint $F_{M}\dashv U_{M}$ defined on objects as
$F_{M}(C) = (MC, \mu_C)$ and on morphisms as $F_M(f) = M(f)$.
The adjunction $F_{M}\dashv U_{M}$ induces, as described above,
the same monad $M$ on $\mathcal{C}$ we started with.


The following is a key property of the Eilenberg--Moore construction
(see, e.g., \cite[VI.3.1]{maclane98book}).

\begin{proposition}[Comparison Functor]\label{prop:key-property-comparison}
Let $F\dashv U$ be an adjoint pair with unit $\eta$ and counit $\epsilon$. Let $M=(M,\eta^M, \mu^M)$ be the induced monad and $\EM(M)$ the associated Eilenberg-Moore category.
There is a functor $K\colon \mathcal{D}\to \EM(M)$,
called the \emph{comparison functor}, defined as follows:
\[
K(D) = (U D, U\epsilon_D)\qquad K(f) = U(f).
\]
\end{proposition}

With the above definitions in place, one easily derives the following proposition
(see, e.g.~\cite[VI.3.1]{maclane98book}).

\begin{proposition}\label{prop:useful-fact-EM}
Fix, as described above, an adjunction $F\dashv U$, the associated monad $M$,
the $\EM(M)$ category, the adjunction $F_M\dashv U_M$ and the comparison functor
$K\colon \mathcal{D}\to \EM(M)$.
Then it holds that: 
\[
F_M = K F \qquad \textnormal{ and } \qquad U = U_M K.
\]  
\end{proposition}


This work relies on a strict notion of monadicity,
as in, e.g.,~\cite[VI.7]{maclane98book}.

\begin{definition}[Strictly Monadic Functor]\label{def:strict-monadic-functor}
We say that a functor $U\colon \mathcal{D}\to \mathcal{C}$ is \emph{strictly monadic}
if $U$ has a left adjoint $F\dashv U$ 
and, denoting with $M=(M,\eta^M, \mu^M)$ the induced monad,
the comparison functor $K\colon \mathcal{D}\to \EM(M)$ is an isomorphism of categories.
\end{definition}

The functor $U_M\colon \EM(M)\to \mathcal{C}$ is, trivially, strictly monadic.
For the following fact see, e.g.,~\cite[3.3.2]{BarrWells2005}.

\begin{proposition}\label{iso-reflection}
Every strictly monadic functor $U\colon\mathcal{D}\to \mathcal{C}$ reflects isomorphisms,
i.e., if $U(f)$ is an isomorphism in $\mathcal{C}$
then $f$ is an isomorphism in $\mathcal{D}$.
\end{proposition}

\subsection{Monad Morphisms and Restrictions} \label{sec:monad-background}
The material presented in this section is folklore.
For instance, the notion of monad morphism between monads on a given category
$\mathcal{C}$ is developed in \cite[\S 3]{BarrWells2005}.
In this work, however,
we require monad morphisms between monads on different categories. 
This theory is naturally developed at the level of 2-categories
(see, e.g., \cite{STREET1972149}),
but we only need an ordinary 1-categorical treatment.

An \emph{embedding of categories}
is a functor which is full, faithful and injective on objects.

\begin{definition}[Monad Morphism]\label{def:monad_morphism}
Let $\mathcal{C}$ and $\mathcal{D}$ be categories, 
$(N,\eta^N,\mu^N)$ a monad on $\mathcal{C}$ and  
$(M,\eta^M,\mu^M)$ a monad on $\mathcal{D}$. 
A pair $(G,\gamma)$, with
$G\colon\mathcal{C}\to \mathcal{D}$
a functor and
$\gamma\colon MG \Rightarrow GN$
a natural transformation, is called a \emph{monad morphism}
$(G,\gamma) \colon (\mathcal{C},N)\to (\mathcal{D},M)$ if it satisfies:
\begin{enumerate}
\item \emph{(Unit compatibility)}
$
\gamma\circ (\eta^M G) = G(\eta^N) $, i.e., 
\[
\gamma_C\circ \eta^M_{GC} = G(\eta^N_C)\colon GC\to G(NC)
\tag{for all $C\in\mathcal{C}$}
\]
\item \emph{(Multiplication compatibility)}
$
\gamma \circ (\mu^M G) = (G\mu^N) \circ (\gamma N) \circ (M\gamma)
$, i.e., 
\[
  \gamma_C \circ \mu^M_{GC}
  =
  G(\mu^N_C) \circ \gamma_{NC} \circ (M\gamma_C)
  \colon MMGC \to GNC 
\tag{for all $C\in\mathcal{C}$}
\]
\end{enumerate}
\end{definition}

The next two propositions link the definition of monad morphism, as above, with the Eilenberg--Moore construction. 

\begin{proposition}\label{prop:from_morphism_to_functor}
Let $(G,\gamma)\colon (\mathcal{C},N)\to ( \mathcal{D},M)$ be a monad morphism.
Then there is a canonical functor $\Gamma_{\gamma}\colon \EM(N)\to \EM(M)$
defined:
\begin{itemize}
\item on objects $(C,\alpha_C\colon NC\to C)$ as:
\[
\Gamma_{\gamma}(C,\alpha_C) := (GC, G\alpha_C\circ \gamma_C \colon MGC \to GC)
\]
\item on morphisms of $N$-algebras $f\colon (C,\alpha_C)\to (C',\alpha_{C'})$ as:
\[
\Gamma_{\gamma}(f) := G(f)\colon GC\to GC'
\]
\end{itemize}
such that $\Gamma_{\gamma}$ lifts $G$, i.e., $U_M\circ  \Gamma_{\gamma}  = G\circ U_N $,
\[\begin{tikzcd}
	{\mathcal{C}} & {\EM(N)} \\
	{\mathcal{D}} & {\EM(M)}
	\arrow["G"', from=1-1, to=2-1]
	\arrow["{U_N}"', from=1-2, to=1-1]
	\arrow["\Gamma_\gamma", from=1-2, to=2-2]
	\arrow["{U_M}"', from=2-2, to=2-1]
\end{tikzcd}\]
\end{proposition}

\begin{proposition}\label{prop:from_functor_to_morphism}
Let $\Gamma\colon\EM(N)\to \EM(M)$ be a functor lifting
$G \colon \mathcal{C}\to\mathcal{D}$,
i.e., such that $U_M\circ  \Gamma  = G\circ U_N $.
Denote with $\xi_C\colon M(GNC)\to GNC$ the $M$-algebra structure of $\Gamma(F_N(C))$,
for $C\in\mathcal{C}$:
\[
\begin{array}{rclr}
\Gamma(F_N(C))& = & \Big( U_M(\Gamma(F_N(C))), \xi_C\Big)\qquad & \text{$M$-algebra with structure $\xi_C$}\\ 
& = & \Big( G(U_N (F_N(C))), \xi_C\Big)\qquad & \text{lifting: $U_M\circ  \Gamma  = G\circ U_N$}\\ 
& = & \big( G N (C), \xi_C\big)\qquad & \text{definition: $N=U_N F_N$}\\
\end{array}
\]
Then there is a canonical natural transformation
$\widehat{\gamma}^{\Gamma}\colon F_M G \Rightarrow \Gamma F_N$ defined,
for every $C\in\mathcal{C}$, by the equation:
\[
  U_M(\widehat{\gamma}^{\Gamma}_C) := \xi_C\circ M(G(\eta^N_C))
  \qquad \qquad
  \widehat\gamma^\Gamma_C \colon (MGC, \mu^M_{GC}) \to (GNC, \xi_C).
\]
Furthermore, the underling natural transformation
\[\gamma^{\Gamma}:=U_M(\widehat{\gamma^\Gamma})\colon MG\Rightarrow GN 
\]
gives a canonical monad morphism $(G,\gamma^\Gamma)$,
i.e., $\gamma^{\Gamma}$ satisfies the unit and multiplication compatibility equations of
Definition~\ref{def:monad_morphism}.
\end{proposition}



\begin{proposition}\label{prop:construction-inversion}
The two constructions $\gamma\mapsto \Gamma_\gamma$ (Proposition~\ref{prop:from_morphism_to_functor}) and $\Gamma \mapsto \gamma^{\Gamma}$ (Proposition~\ref{prop:from_functor_to_morphism}) are inverse to each other, i.e.:
\begin{enumerate}
	\item $\gamma = \gamma^{\Gamma_{\gamma}}$, for every monad morphism $(G,\gamma)$, 
	\item $\Gamma = \Gamma_{\gamma^\Gamma}$, for every functor $\Gamma\colon \EM(N)\to \EM(M)$ lifting $G$ (i.e., $U_M \Gamma = G U_N$).
\end{enumerate}
\end{proposition}

\begin{remark}
The three propositions above, for the particular case when
$\mathcal{C}= \mathcal{D}$
(i.e., when $M$ and $N$ are monads on the same category)
appear as~\cite[Theorem~3.6.3]{BarrWells2005}. 
In fact, the latter further shows that the one-to-one correspondence
$\gamma \leftrightarrow \Gamma_\gamma$ also preserves composition.
We will not need this additional fact. 
\end{remark}

Consider now a monad morphism $(G,\gamma)$ such that $\gamma$ is a natural isomorphism.  
By applying Proposition~\ref{prop:from_morphism_to_functor},
we obtain a canonical functor $\Gamma_{\gamma}\colon \EM(N)\to \EM(N)$ lifting $G$, 
and by Proposition~\ref{prop:from_functor_to_morphism} a canonical natural transformation 
$\widehat{\gamma}\colon F_M M\Rightarrow \Gamma_{\gamma} F_N$
such that, by Proposition~\ref{prop:construction-inversion},
$U_M(\widehat{\gamma})$ is $\gamma$.  
It is not obvious, \emph{a priori},
that $\widehat{\gamma}$ is itself a natural isomorphism. 
This is because the existence on an inverse of $\gamma$ in $\mathcal{D}$
does not necessarily imply the existence of an inverse of $\widehat{\gamma}$ in $\EM(M)$. 
This requires showing that $\gamma^{-1}$ lives in $\EM(M)$ as well, i.e.,
that it satisfies the constraints of $M$-algebra morphisms. 
The following direct consequence of Proposition~\ref{iso-reflection}
states that this is always the case.

\begin{proposition}\label{prop:isomoprhism-preservation}
Let $(G,\gamma)$ be a monad morphism with $\gamma$ a natural isomorphism.
Then the canonical
$\widehat{\gamma}\colon F_M M\Rightarrow \Gamma_{\gamma} F_N$
is a natural isomorphism. 
\end{proposition}

\begin{definition}\label{def:monad-restriction}[Monad Restriction]
A \emph{restriction of monads} is a monad morphism
$(G,\gamma)\colon (\mathcal{C}, N) \to (\mathcal{D}, M)$
such that $G$ is an embedding of categories and $\gamma$ a natural isomorphism.  
\end{definition}

The next fact states that a restriction $(G,\gamma)$,
if it exists, is unique up to isomorphism.
Hence, when a restriction exists,
we can talk about \emph{the} restriction of $N$ along the embedding $G$.
And we will say that $M$ is a monad restriction of $N$ along $G$ if \emph{the}
restriction $(G,\gamma)$ exists. 

\begin{proposition}\label{prop:restriction_is_unique}
Let $G\colon \mathcal{C}\to\mathcal{D}$ be an embedding of categories.
Any two restrictions $(G,\gamma)$ and $(G,\delta)$,
are equal up to isomorphism, i.e., there is a natural isomorphism
$\iota \colon GN\Rightarrow GN$ such that $\iota\circ \gamma=\delta$, defined as $\iota_C =\delta_C \circ \gamma_C^{-1}$. 
\end{proposition}


We will use this corollary of the preceding propositions.

\begin{corollary}\label{cor:restriction_lemma_to_be_used}
Let $G\colon\mathcal{C}\to\mathcal{D}$ be a an embedding of categories, $N$ a monad on $\mathcal{C}$ and $M$ a monad on $\mathcal{D}$. The following are equivalent:
\begin{enumerate}
\item $N$ is a restriction of $M$ along $G$, i.e., there is a monad restriction $(G,\gamma)$ (see Definition~\ref{def:monad-restriction}),
\item there is a functor $\Gamma\colon EM(N)\to EM(M)$ that lifts $G$ and such that the associated  
canonical natural transformation $\widehat{\gamma}^{\Gamma}\colon F_M G \Rightarrow \Gamma F_N$ (as in Proposition~\ref{prop:from_functor_to_morphism}) is a natural isomorphism.
\end{enumerate}
\end{corollary}


\subsection{Categories of Models of Horn Theories}\label{sec:horn-background}
In the previous subsections we have considered monads on abstract categories.
We now specialise the discussion to categories of models of Horn theories.
We refer to \cite{hodges93book} and \cite{Adamek_Rosicky_1994} as standard references.

A signature $\Sig$ is a possibly infinite collection of relation symbols
$R\in \Sig$ and function symbols $\sigma \in \Sig$.
Throughout this work we assume that all symbols $R,\sigma\in \Sig$ have finite arity. 
We reserve the letters $\Pi$ and $\Sigma$ for signatures consisting uniquely of
relation symbols and function symbols, respectively.

An $\Sig$-structure is a pair:
\[
  A = (|A|, \mathcal{I}^A)
  \qquad \qquad \textnormal{ where }
  \mathcal{I}^A = \{R^A\}_{R\in \Sig}\cup \{\sigma^A\}_{\sigma \in \Sig}
\]
i.e., a carrier set $|A|$ together with interpretations $R^A\subseteq {|A|}^{ar(R)}$
and $\sigma^A\colon {|A|}^{ar(\sigma)}\to {|A|}$ of all relations and function symbols
in $\Sig$.  Given $\Sig$-structures $A$ and $B$,
an $\Sig$-homomorphism is a map $h\colon |A|\to |B|$ preserving
all symbols in $\Sig$, i.e.,
\(
(a_1,\dots, a_n)\in R^A \Rightarrow \big(h(a_1),\dots, h(a_n)\big)\in R^B
\)
and
\(
h\big( \sigma^A(a_1,\dots, a_n) \big) = \sigma^B\big( h(a_1),\dots, h(a_n)\big)
\).
We denote with $\Str(\Sig)$ the category of $\Sig$-structures and their homomorphisms.

We are interested in full subcategories of $\Str(\Sig)$ defined by
infinitary Horn sentences in the signature~$\Sig$.
Such Horn sentences are of the form
\[
  \forall \vec{x}. \ \displaystyle \Big( \bigwedge_{i\in I} \atom_i(\vec{x}) \Big)
  \Rightarrow
  \atombis(\vec{x})
\]

\noindent
allowing for infinitely many variables and formulas in the antecedent of the implication,
and where the formulas $\atom_i(\vec{x}),\atombis(\vec{x})$ are atomic formulas
in the signature $\Sig$ with variables from $\vec{x}$
(an atomic formula in $\Sig$
is either a relation $R(t_1,\dots,t_{ar(R)})$ with $R \in \Sig$
or an equality $t_1 = t_2$, where the $t_i$'s are terms in $\Sig$).
If a Horn sentence $H$ involves only finitely many variables and finitely many
formulas in the antecedent, i.e., if it belongs to ordinary first order logic,
then we refer to it as \emph{finitary}.

Satisfiability between
$\Sig$-structures $A$ and Horn sentences $H$
is defined as usual and denoted with $A \models H$.
If $\Horn$ is a set of Horn sentences,
we denote with $\Mod_{\Sig}(\Horn)$ the full subcategory of $\Str(\Sig)$
whose objects satisfy all formulas in $\Horn$.
Note that $\Str(\Sig) = \Mod_{\Sig}(\emptyset)$.

\begin{theorem}
\label{thm:horn:adj}
Let $\Sig, \Sigbis$ be signatures with $\Sigbis\supseteq \Sig$,
i.e.,  $\Sigbis$ is an extension of $\Sig$.
Let $\Horn$ be a set of Horn sentences in $\Sig$
and $\Hornbis$ be a set of Horn sentences in $\Sigbis$
with $\Hornbis \supseteq \Horn$,
i.e., $\Hornbis$ is an extension of $\Horn$.
Then:
\begin{enumerate}
\item
\label{item:horn:adj:forget}
There is a forgetful functor
$U\colon \Mod_{\Sigbis}(\Hornbis) \to \Mod_{\Sig}(\Horn)$
that forgets the additional symbols in $\Sigbis$
and the additional Horn sentences in $\Hornbis$.
Furthermore, if $\Sig=\Sigbis$ then the functor $U$ is an embedding of
$\Mod_{\Sigbis}(\Hornbis)$ as a full subcategory of $\Mod_{\Sig}(\Horn)$. 

\item
\label{item:horn:adj:adj}
$U$ has a left adjoint $F\dashv U$, so there is a monad on
$\Mod_{\Sig}(\Horn)$ induced by the adjunction $F\dashv U$.
\end{enumerate}
\end{theorem}


Rather than providing a concrete definition of the left adjoint $F$,
we will work with its abstract properties described in the next
Subsection~\ref{sec:presentations}.

\begin{remark}\label{example:non-strictly-monadic}
The forgetful functor $U$ of Theorem \ref{thm:horn:adj} 
is, in general, not strictly monadic
(Definition~\ref{def:strict-monadic-functor}).
For example, taking $\Sig=\Horn=\emptyset$, 
the category $\Mod_{\Sig}(\Horn)$ is the category $\Set$.
And by taking $\Sigbis=\{ \leq \}$ and $\Hornbis$ to be the Horn
implications defining partial orders, 
the category $\Mod_{\Sigbis}(\Hornbis)$ is the category $\mathbf{Pos}$
of partial orders and monotone maps. In this case $U$ is not strictly monadic.
\end{remark}



\subsubsection{Presentations of Structures}\label{sec:presentations}

We now review well-known material on presentations of structures
by generators and relations.
A standard reference is \cite[\S 9.2]{hodges93book}. 
The following notions are 
from~\cite[1.2.2, 1.4.1 and~\S 9.2]{hodges93book}
(see also \cite[\S 3.10]{Adamek_Rosicky_1994}).

\begin{definition}
An injective $\Sig$-homomorphism $h\colon A\to B$ is an \emph{embedding} if it
\emph{reflects} all relation symbols $R \in \Sig$, i.e.,
\(
\big(h(a_1),\dots, h(a_n)\big)\in R^B
\Rightarrow
(a_1,\dots, a_n)\in R^A
\).
We say that $A$ is a \emph{substructure} of $B$ if ${|A|} \sle {|B|}$,
and if moreover the inclusion function $|A| \into |B|$
is an embedding $A \to B$.
\end{definition}

\begin{definition}
Given a tuple of elements $\vec a$ from an $\Sig$-structure $A$,
we say that $\vec a$ \emph{generates} $A$ if 
the smallest substructure of $A$ that contains $\vec a$ is $A$ itself.
\end{definition}

\begin{definition}[Presentations]
An $\Sig$--\emph{presentation} is a pair $(\vec x, \vec\atom)$
where $\vec x = (x_j)_{j \in J}$ is a possibly infinite sequence of variables,
and where
$\vec\atom = (\atom_i)_{i \in I}$
is a possibly infinite sequence of atomic formulas in $\Sig$
with variables from $\vec x$.
\end{definition}


\begin{definition}
\label{def:horn:pres}
Let $\cat K \subseteq \Str(\Sig)$ be a class of $\Sig$-structures.
Fix a presentation
$((x_j)_{j \in J}, (\atom_i)_{i \in I})$ 
and a structure $A \in \cat K$.
Let $\vec a = (a_j)_{j \in J}$
be a sequence of elements in $A$, one for each $x_j$, referred to as \emph{parameters}.
We say that $(\vec x, \vec\atom)$
\emph{presents} the expanded structure $(A,\vec a)$ in $\cat K$
when the three following conditions hold:
\begin{enumerate}
\item
\label{item:horn:pres:gen}
The $\Sig$-structure $A$ is generated by $\vec a$.

\item
\label{item:horn:pres:sat}
$(A,\vec a)\models \atom_i$ for each $i \in I$,
where this notation means that each variable $x_j$ in $\alpha_i$
is interpreted by the parameter $a_j$.

\item
\label{item:horn:pres:univ}
For all $B \in \cat K$ and elements $\vec b=( b_j)_{j\in J}$ of $B$,
if $(B,\vec b)\models \atom_i$ holds for all  $i \in I$,
then there is a homomorphism of expanded $\Sig$-structures
$h\colon (A, \vec a) \to (B,\vec b)$,
that is,
an $\Sig$-homomorphism $h \colon A\to B$ such that $h(a_j)=b_j$ for all $j\in J$.
\end{enumerate}
\end{definition}

\begin{remark}
\label{remark:iso-from-presentation}
In item~(\ref{item:horn:pres:univ}) of Definition~\ref{def:horn:pres},
since $A$ is generated by $\vec a$,
the homomorphism of expanded structures $h$ is unique
whenever it exists.
%
%
As a consequence,
if $(A,\vec a)$ and $(B,\vec b)$
are both presented by $(\vec x,\vec\atom)$ in $\cat K$,
then the unique $h \colon (A, \vec a) \to (B,\vec b)$ 
is an isomorphism of expanded structures.
\end{remark}

For a proof of the following useful fact see~\cite[Lemma~9.2.1]{hodges93book}.

\begin{lemma}
\label{lem:horn:pres}
Let $\cat K$ be a class of $\Sig$-structures
and consider a presentation 
$((x_j)_{j \in J}, (\atom_i)_{i \in I})$.
The following are equivalent for each expanded structure $(A,\vec a)$
with $A \in \cat K$:
\begin{enumerate}
\item
\label{item:horn:pres}
$(\vec x, \vec\atom)$ presents $(A,\vec a)$ in $\cat K$.

\item
\label{item:horn:char}
$\vec a$ generates the $\Sig$-structure $A$
and for every atomic formula $\atombis(\vec x)$ in $\Sig$,
we have:
\begin{center}
\(
  (A,\vec a) \models \atombis
  \qquad \Longleftrightarrow \qquad
  \cat K \models \forall \vec x.
    \left(\bigwedge_{i\in I} \atom_i\right) \Rightarrow \atombis
\).
\end{center}
\end{enumerate}
\end{lemma}

We apply this machinery to the setting of Theorem~\ref{thm:horn:adj}, where we have signatures $\Sigbis\supseteq \Sig$, Horn theories $\Hornbis\supseteq \Horn$
and the adjunction $F\dashv U \colon \Mod_\Sig(\Horn) \to \Mod_\Sigbis(\Hornbis)$.
\begin{definition}\label{def:presentation-operation}
    Given any $A\in \Mod_\Sig(\Horn)$ we define the presentation $\Pres(A)= (\vec{x},\Phi_A)$, where:
    \begin{itemize}
        \item $\vec{x} = ( x_a )_{a\in A}$, i.e., there is one variable $x_a$ for each $a\in A$,
        \item $\Phi_A$ is the set of all atomic $\mathcal{S}$-formulas $\atom(\vec x)$ such that $A,\vec{a}\models \atom(\vec{x})$, i.e., that are satisfied in $A$ when interpreting $x_a$ by $a$.
    \end{itemize}
\end{definition}

In model theory, the set of formulas $\Phi_A$ of $\Pres(A)$ is
often referred to as the \emph{positive diagram of $A$}
(see, e.g.,~\cite[Remark 5.33]{Adamek_Rosicky_1994} and \cite{hodges93book}).
Note that $\Pres(A)$ is a presentation in the signature $\Sig$
(because all formulas in $\Phi_A$ are atomic $\Sig$-formulas)
and therefore, \emph{a fortiori}, also a presentation in the extended signature $\Sigbis$. 
In the following result, which provides a logical characterisation of free objects, we use $\Pres(A)$ as a presentation in the extended signature $\Sigbis$. 

\begin{proposition}
\label{prop:main-fact-presentations}
Let $\Horn$ and $\Hornbis$ be Horn theories in signatures
$\Sig$ and $\Sigbis$, respectively, with $\Sig \sle \Sigbis$ and $\Horn \sle \Hornbis$.
Let $F\dashv U \colon \Mod_\Sig(\Horn) \to \Mod_\Sigbis(\Hornbis)$
be the corresponding adjoint pair, as in Theorem~\ref{thm:horn:adj}, with unit $\eta$. 
Fix any $A\in \Mod_\Sig (\Horn)$
and corresponding presentation $\Pres(A)= (\vec{x},\Phi_A)$. Then it holds that:
\begin{enumerate}
\item
\label{item:main-fact-presentations:pres}
The expanded structure
$\big( F A, (\eta_A(a) \mid a\in A) \big)$
is presented by $(\vec{x},\Phi_A)$ in  $\Mod_\Sigbis(\Hornbis)$.

\item
\label{item:main-fact-presentations:bridge}
For each atomic formula $\atombis(\vec{x})$ in the extended signature $\Sigbis$, 
\[
  \big( F A, (\eta_A(a) \mid a\in A) \big) \models \atombis(\vec{x})
  \qquad \Longleftrightarrow \qquad
  \Hornbis \models \forall{\vec{x}}. (\bigwedge \Phi_A) \Rightarrow \atombis.
\]
\end{enumerate}
\end{proposition}

\begin{proof}
The universal property of $F A$, as a free object, is used to establish the first point. 
The second point is a direct consequence of Lemma~\ref{lem:horn:pres}.
\end{proof}

\begin{remark}\label{remark:eta-parameters}
In Proposition \ref{prop:main-fact-presentations},
$\eta_A(a)$ is viewed as an element of the structure $F(A)$. 
Note that, formally, $\eta_A\colon A \to UF(A)$, and therefore $\eta_A(a)$,
is an element of $UF(A)$.
However, by definition of the forgetful functor $U$,
the two structures $F(A)\in \Mod_\Sigbis(\Hornbis)$
and $UF(A)\in \Mod_\Sig(\Horn)$ have the same carrier,
i.e., $|F(A)|= |UF(A)|$. Hence $\eta_A(a)$ is an element of $F(A)$.
\end{remark}

\begin{lemma}
\label{lem:pres:unique-iso}
Let $\Horn$ and $\Hornbis$ be Horn theories in signatures
$\Sig$ and $\Sigbis$, respectively, with $\Sig \sle \Sigbis$ and $\Horn \sle \Hornbis$.
Let $F\dashv U \colon \Mod_\Sig(\Horn) \to \Mod_\Sigbis(\Hornbis)$
be the corresponding adjoint pair, as in Theorem~\ref{thm:horn:adj}, with unit $\eta$. 
Fix $A\in  \Mod_\Sig(\Horn)$
and let $\vec{x}=(x_a \mid a\in A)$ be a collection of variables, one for each $a\in A$.
Let furthermore
$B\in \Mod_\Sigbis(\Hornbis)$ with parameters $ ( b_a \mid a \in A)$, such that:
\begin{enumerate}
\item
\label{item:pres:unique-iso:gen}
$B$ is generated by the parameters $ ( b_a \mid a \in A)$, and

\item
\label{item:pres:unique-iso:equiv}
for every atomic formula $\atombis(\vec{x})$ in the signature $\Sigbis$,
\[
  (F A, (\eta_A(a) \mid a\in A) ) \models \atombis(\vec{x})
  \qquad \Longleftrightarrow \qquad
  (B, ( b_a \mid a \in A)) \models \atombis(\vec{x}).
\]
\end{enumerate}

\noindent
Then there exists a unique homomorphism of expanded structures
$h\colon (F A, (\eta_A(a) \mid a\in A) \big)\to (B,\vec{b})$,
and $h$ is an isomorphism.
\end{lemma}



\section{Quantitative Algebras}\label{sec:quantitative-algebra-background}
Quantitative algebra, as originally introduced in \cite{DBLP:conf/lics/MardarePP16} (see also \cite{DBLP:journals/jacm/BacciMPP26}), studies algebraic theories over the category of metric spaces and nonexpansive maps.
The apparatus allows to explore the interactions between the metric structure and the algebraic structure.
In a similar but more general way, some works \cite{MioEA24,entics:16876} have considered
algebraic theories over the category of fuzzy relations. In this paper, inspired by both \cite{FordEA21} and \cite{MioEA24}, we adopt an even more general definition. The \emph{base} category, such as metric or fuzzy relation spaces, 
is assumed to be $\Mod_\Pi(\mathbb{B})$, the models of a Horn theory $\mathbb{B}$ over a purely relational (i.e., not having any constant or functions symbols) signature $\Pi$. The signature $\Pi$ can be infinite
but all relations $R\in \Pi$ must have finite arity.

\begin{example}\label{example:base-theory-fuzzy-relations}
The category $\FRel$ of fuzzy relations and nonexpansive maps is (isomorphic to) $\Mod_\Pi(\mathbb{B})$,
where $\Pi := {\{ =_\epsilon \}}_{\epsilon\in [0,1]}$ is an infinite set of binary relations indexed by real numbers $0 \leq \epsilon \leq 1$,
and  $\mathbb{B}$ is the Horn theory consisting of the following sentences:
\begin{itemize}
    \item (Monotonicity) $\forall x,y. (x=_\epsilon y) \Rightarrow  (x=_\delta y)$, for each $\epsilon, \delta\in [0,1]$ with $\epsilon <\delta$, 
    \item (Order Completeness) $\forall x,y. \displaystyle  \Big( \bigwedge_{\epsilon \in [0,1]\cap\mathbb{Q}, \epsilon > \delta}  x=_\epsilon y \Big)\Rightarrow x=_\delta y$,  for each $\delta \in [0,1]$.
\end{itemize}
Note that each instance of the \emph{order completeness}
implication has countably many premises. 
It can be shown that the category $\FRel$ requires infinitary formulas to be axiomatised. This follows, for example,
from the fact that $\FRel$ is locally presentable but not locally finitely presentable.
The category $\Met$ of $1$-bounded metric spaces and nonexpansive maps is obtained by adding to $\mathbb{B}$ additional Horn implications expressing the axioms of metric spaces, such as the triangle inequality (see, e.g.~\cite[\S 2.3]{MioEA24}). 
\end{example}

The algebraic part of the theory, on top of the base category, is formalised by considering a signature $\Sigma$ consisting only of function symbols  (i.e., no relations) of finite arity (constants are functions of arity $0$)
and a Horn theory $\mathbb{Q}$ over the extended signature $(\Sigma+\Pi)$ consisting of \emph{basic quantitative equations}, defined as follows.
\begin{definition}\label{def:basic-equation}
Let $(\Pi,\mathbb{B})$ be a base theory. Let $\Sigma$ consist only of constants and function symbols of finite arity. A \emph{basic quantitative equation} (BQE) 
is a Horn implication of the form:
\[
\forall \vec{x}. \Big( \bigwedge_{i \in I} \atom_i(\vec{x}) \Big) \Rightarrow \atomext(\vec{x})
\]
such that the function symbols in $\Sigma$ never appear in the antecedent. Equivalently:
\begin{itemize}
    \item the formulas $\atom_i(\vec{x})$ are atomic $\Pi$-formulas using variables from $\vec{x}$. So $\atom_i$ is either of the form $x_1=x_2$ or $R(x_1,\dots, x_n)$, with $R\in\Pi$,
\item the formula $\atomext(\vec{x})$ is an atomic $(\Pi +  \Sigma)$-formula using variables from $\vec{x}$. So $\atomext(\vec{x})$ is of the form $R(s_1(\vec{x}),\dots, s_n(\vec{x}))$, where $R$ is either a relation symbol $R\in \Pi$ or the equality symbol $(=)$, 
and $s_1,\dots, s_n$ are $\Sigma$--terms with variables from $\vec{x}$.
\end{itemize}
As usual, if the conjunction in the antecedent is empty, we just write $\forall \vec{x}. \atomext(\vec{x})$.
\end{definition}
\begin{remark}
The use of the adjective ``basic'' in Definition~\ref{example:base-theory-fuzzy-relations} follows the terminology adopted in the first paper on quantitative algebra~\cite{DBLP:conf/lics/MardarePP16}, where arbitrary Horn implications over the signature $\Pi+\Sigma$
are called \emph{quantitative inferences}, the subclass having no function symbols in the antecedent, as in Definition~\ref{def:basic-equation}, are called \emph{basic quantitative inferences},
and the further subclass of the form $\forall \vec{x}. \atomext(\vec{x})$ are called \emph{quantitative equations}. In later works
the theory has developed with a focus on \emph{basic quantitative inferences} (see e.g., \cite{DBLP:conf/lics/MardarePP17,FordEA21,MioEA24}),
rather than \emph{quantitative inferences}. Accordingly, in this work we only work with \emph{basic quantitative equations} in the sense of Definition~\ref{def:basic-equation}.
Somewhat confusingly,\footnote{See also Remark 2.12 in~\cite{entics:16876}.}\
in~\cite{entics:16876} the term \emph{quantitative equation} refers to a special class of basic quantitative equations (as in Definition~\ref{def:basic-equation}),
which we will refer to as \emph{antecedent complete BQE} (Definition~\ref{def:antecedente-complete} of Section~\ref{sec:compact-definition}). 
\end{remark}

\begin{definition}
A \emph{quantitative equational theory} $T$ is given by a base theory
$(\Pi,\mathbb{B})$ and an algebraic theory $(\Sigma,\mathbb{Q})$, 
with the signature $\Sigma$ consisting only of constants and function symbols and $\mathbb{Q}$ consisting of basic quantitative equations, as defined above. So we can write $T=\langle (\Pi,\mathbb{B}),  (\Sigma,\mathbb{Q})\rangle$. 
The category of models of the theory $T$, called \emph{quantitative algebras} (of the theory $T$), is $\Mod_{\Pi+\Sigma}(\mathbb{B}+\mathbb{Q})$.
\end{definition}
\begin{remark}
    Note that there is nothing necessarily ``quantitative'' in this definition. We adopt this adjective 
    merely to adhere with previous literature on this line of research, starting with~\cite{DBLP:conf/lics/MardarePP16}, where the focus 
    was on variants of the base theory of Example~\ref{example:base-theory-fuzzy-relations}. 
\end{remark}
\begin{example}\label{example:quantitative-convex-algebras}
Consider the base theory $(\Pi,\mathbb{B})$ of fuzzy relations,
from Example~\ref{example:base-theory-fuzzy-relations}.
Let $\Sigma = {\{+_p\}}_{p\in (0,1)}$ consist of a (uncountable) collection
of binary function symbols indexed by reals $0 < p < 1$.
Consider the set $\mathbb{Q}$ of basic quantitative equations given by:
\begin{itemize}
\item Axioms of convex algebras:
\[
\forall x,y. (x+_p x = x) \qquad \forall x,y.( x+_p y = y +_{1-p}x) \qquad \forall x,y,z. \big(x +_p (y +_q z) = (x+_{pq} y ) +_{\frac{(1-p)q}{1-pq}} z\big) 
\]
\item Interpolative barycentric axiom, for all $\epsilon,\delta \in [0,1]$ and $0<p<1$:
\[
\forall x,y,w,z. \big( x=_\epsilon w \ \wedge \ y=_\delta z \big) \Rightarrow x+_p y =_{\epsilon p + \delta(1-p)} w+_p z.
\]
\end{itemize}
The models of
$T=\langle (\Pi,\mathbb{B}), (\Sigma,\mathbb{Q})\rangle$,
i.e., the objects in $\Mod_{\Sigma+\Pi}(\mathbb{B}+\mathbb{Q})$,
are the \emph{quantitative convex algebras} of~\cite{entics:16876},
or equivalently structures
$(A,d\colon A^2\rightarrow [0,1], {\{+^A_p\}}_{p\in (0,1)})$
that are both fuzzy relations $(A,d)$ and convex algebras
$(A, {\{+^A_p\}}_{p\in (0,1)})$,
with the two being compatible in the sense of the interpolative barycentric axiom,
by interpreting $a=_\epsilon a'$ as $d(a,a')\leq \epsilon$.
Quantitative convex algebras over metric spaces where 
earlier introduced in~\cite{DBLP:conf/lics/MardarePP16},
and can be analogously formalised
by taking $\mathbb{B}$ to be the basic theory axiomatizing $1$-bounded metric spaces
(see Example~\ref{example:base-theory-fuzzy-relations}).
\end{example}

The following is a key result regarding quantitative equational theories.

\begin{theorem}%
\label{thm:stict-monadicity-quantitativetheory}
Let
$T = \langle (\Pi,\mathbb{B}),  (\Sigma,\mathbb{Q})\rangle$
be a quantitative equational theory. 
There is a forgetful functor
$U\colon \Mod_{\Pi+\Sigma}(\mathbb{B}+\mathbb{Q}) \rightarrow \Mod_\Pi(\mathbb{B})$,
from the category of models of $T$ to the base category,
that forgets the additional function symbols in $\Sigma$
and the Horn implications in $\mathbb{Q}$. Furthermore:
\begin{enumerate}
\item\label{item:stict-monadicity-quantitativetheory:ladj}
$U$ has a left adjoint $F\dashv U$,
which constructs free quantitative algebras generated by objects in the base category.

\item\label{item:stict-monadicity-quantitativetheory:monadic}
$U$ is strictly monadic (see Definition~\ref{def:strict-monadic-functor}), i.e., by denoting with $T$ the monad induced by $F\dashv U$,
the comparison functor
$K\colon \Mod_{\Pi+\Sigma}(\mathbb{B}+\mathbb{Q})\to \EM(T)$
is an isomorphism of categories.

\item\label{item:reflection-of-iso}
$U$ reflects isomorphisms, that is,
if $U(f)$ is an isomorphism then $f$ is an isomorphism.
\end{enumerate}
\end{theorem}

\begin{proof}
Item~(\ref{item:stict-monadicity-quantitativetheory:ladj})
is an instance of Theorem~\ref{thm:horn:adj}.
Item~(\ref{item:stict-monadicity-quantitativetheory:monadic}),
with differences depending on the exact mathematical
setting, has appeared in several works on quantitative algebra 
(\cite[Thm~4.13]{FordEA21}, \cite[Thm~2.17]{adamek22}, \cite[Thm~6.3]{MioEA24}).
We offer a proof in \S\ref{sec:proof:monadicity}.
Item~(\ref{item:reflection-of-iso}) follows
from~(\ref{item:stict-monadicity-quantitativetheory:monadic})
using Proposition~\ref{iso-reflection}.
\end{proof}

\begin{remark}
The restrictions imposed
on $\Sigma$ (purely functional signature) and $\mathbb{Q}$
(consisting of basic quantitative equations,
as opposed to arbitrary Horn implications in the signature $\Pi+\Sigma$)
are used for proving item~(\ref{item:stict-monadicity-quantitativetheory:monadic})
of Theorem~\ref{thm:stict-monadicity-quantitativetheory}.
Without appropriate conditions, item~(\ref{item:stict-monadicity-quantitativetheory:monadic}) is generally false 
(cf.\ Remark~\ref{example:non-strictly-monadic}). 
\end{remark}

As a direct consequence of Theorem~\ref{thm:stict-monadicity-quantitativetheory} we have the following.

\begin{corollary}\label{corollary:monad-of-a-theory}
To every quantitative equational theory $T = \langle (\Pi,\mathbb{B}),  (\Sigma,\mathbb{Q})\rangle$,
corresponds a monad on $\Mod_{\Pi}(\mathbb{B})$, also denoted by $T$,
induced by the adjoint pair $F\dashv U$, such that $\Mod_{\Pi+\Sigma}(\mathbb{B}+\mathbb{Q}) \cong \EM(T)$.
\end{corollary}

\begin{example}
Consider the quantitative theory $T=\langle (\Pi,\mathbb{B}),  (\Sigma,\mathbb{Q})\rangle$ of Example~\ref{example:quantitative-convex-algebras} defining \emph{quantitative convex algebras}. 
Recall that, in this example, the base category $\Mod_\Pi(\mathbb{B})$ is (isomorphic to) $\mathbf{FRel}$. 
The corresponding monad $T\colon\FRel\to\FRel$ is the
\emph{Kantorovich Distribution monad}, which takes a fuzzy relation $(X,d)$
and maps it to the fuzzy relation $(D_f(X), K(d))$ of finite distributions on $X$ equipped with the Kantorovich lifting of the fuzzy relation $d$ (see~\cite{entics:16876}). 
If the base theory $(\Pi,\mathbb{B})$ is taken to be that of $1$-bounded metric spaces (see Example~\ref{example:base-theory-fuzzy-relations}), the resulting monad is the Kantorovich Distribution monad $T\colon\Met\to\Met$ studied in
the first article on quantitative algebras~\cite{DBLP:conf/lics/MardarePP16}.  
\end{example}

\begin{remark}
As already mentioned, the definitions presented in this section are inspired by both~\cite{FordEA21} and~\cite{MioEA24}.
As in~\cite{FordEA21}, the base category $\Mod_{\Pi}(\mathbb{B})$ is always symmetric monoidal closed and
therefore it can be enriched over itself.
But as in~\cite{MioEA24}, the resulting monad $T$ is in general not enriched over the base category.
The latter property holds if $\mathbb{Q}$ contains basic quantitative equations enforcing
that every operations $\sigma\in\Sigma$ satisfies the conditions for being a homomorphism in the base category.
See~\cite{MioEA24} for a detailed discussion.  
\end{remark}

\subsection{Proof of the Monadicity Theorem~\ref{thm:stict-monadicity-quantitativetheory}%
(\ref{item:stict-monadicity-quantitativetheory:monadic})}%
\label{sec:proof:monadicity}

Let $T = \langle (\Pi, \mathbb{B}), (\Sigma, \mathbb{Q})\rangle$
be a quantitative equational theory.
Consider the functor
\[
  U\colon \Mod_{\Pi + \Sigma}(\mathbb{B} + \mathbb{Q}) \rightarrow \Mod_\Pi(\mathbb{B})
\]
from the category of models of $\mathbb{B} + \mathbb{Q}$ to the base category,
that forgets the function symbols in~$\Sigma$ and the Horn implications in $\mathbb{Q}$.
By Theorem~\ref{thm:stict-monadicity-quantitativetheory}.(\ref{item:stict-monadicity-quantitativetheory:ladj}),
we know that $U$ has a left adjoint $F \dashv U$.  Write $T$ for the induced monad on $\Mod_{\Pi}(\mathbb{B})$.

We show that $U$ is \emph{strictly} monadic in the sense of~\cite[\S VI.3]{maclane98book},
namely that the comparison functor $K\colon \Mod_{\Pi + \Sigma}(\mathbb{B} + \mathbb{Q}) \to \EM(T)$
is an isomorphism of categories.
Similarly as for other variations of quantitative algebras (\cite{FordEA21,adamek22}),
we apply Beck's (Strict) Monadicity Theorem (\cite[Theorem~VI.7.1]{maclane98book}).
We are thus going to prove the following.

\begin{proposition}%
\label{prop:stict-monadicity-quantitativetheory:ladj}
The right adjoint $U$ creates coequalizers for those parallel pairs $(f,g)$ in $\Mod_{\Pi+\Sigma}(\mathbb{B}+\mathbb{Q})$
such that $(U f, U g)$ has a split coequalizer in $\Mod_\Pi(\mathbb{B})$.
\end{proposition}

Recall from~\cite[\S VI.6]{maclane98book} that a \emph{split} coequalizer of a parallel pair
$f, g \colon a \to b$ is given by morphisms $e \colon b \to c$, $s \colon c \to b$ and $t \colon b \to a$
such that the diagram below commutes
\begin{equation}
\label{eq:splitcoeq}
\begin{tikzcd}
  b
  \arrow{d}{e}
  \arrow{r}{t}
& a
  \arrow{r}{f}
  \arrow{d}{g}
& b
  \arrow{d}{e}
\\
  c
  \arrow{r}{s}
& b
  \arrow{r}{e}
& c
\end{tikzcd}
\end{equation}

\noindent
and moreover both horizontal composites are identities.
Note that $e$ is always a coequalizer of $f$ and $g$.
Moreover, split coequalizers are absolute (i.e.\ preserved by all functors).


The proof of Proposition~\ref{prop:stict-monadicity-quantitativetheory:ladj}
is similar to the proof of the analogous~\cite[Thm~4.13]{FordEA21} and~\cite[Thm~2.17]{adamek22}.
Given a parallel pair $f,g \colon A \to B$ in $\Mod_{\Pi + \Sigma}(\mathbb{B} + \mathbb{Q})$
such that the pair $U f, U g \colon U A \to U B$
has a split coequalizer $e \colon U B \to X$ in $\Mod_{\Pi}(\mathbb{B})$,
we use the fact $e$ is absolute to obtain a (unique) $\Sigma$-structure on $X$,
say $C \in \Str(\Pi + \Sigma)$.
This follows the usual argument for the purely algebraic case
(see, e.g., \cite[Remark~3.4(7)]{Adamek_Rosicky_1994}).

In the purely algebraic case, the induced structure on $C \in \Str(\Pi + \Sigma)$
is ensured to satisfy the required equations,
since ordinary equations are preserved by homomorphisms.
However, the preservation of \emph{quantitative} equations requires a separate argument
(recall that arbitrary Horn sentences are \emph{not} preserved by homomorphisms).
This is provided by the following Lemma~\ref{lem:horn:splitepi}.
The proofs of~\cite[Thm~4.13]{FordEA21} and~\cite[Thm~2.17]{adamek22}
rely on similar arguments.

\begin{lemma}%
\label{lem:horn:splitepi}
Let $e \colon A \to B$ be a homomorphism in $\Str(\Pi + \Sigma)$
such that $e$ is split epic in $\Str(\Pi)$:
there is a $\Pi$-homomorphism $s \colon B \to A$
such that $e \comp s$ is the identity.
If the $(\Pi + \Sigma)$-structure $A$ is a model of $\mathbb{B} + \mathbb{Q}$,
i.e.\ $A \in \Mod_{\Pi + \Sigma}(\mathbb{B} + \mathbb{Q})$,
then $B \in \Mod_{\Pi + \Sigma}(\mathbb{B} + \mathbb{Q})$ as well.
\end{lemma}

\begin{proof}
    Consider a sentence in $\mathbb{B} + \mathbb{Q}$, say
\[
\forall \vec{x}.
  \Big( \bigwedge_{i} \atom_i(\vec{x}) \Big) \Rightarrow \atombis(\vec{x})
\]
where $\atombis$ is an atomic formula in $\Pi + \Sigma$
and the $\atom_i$'s are atomic formulas in $\Pi$
(note that $\atombis \in \Pi$ if the above sentence is in $\mathbb{B}$).
Given $\vec b \in B$ such that $B \models \bigwedge_{i} \atom_i(\vec b)$,
we have $A \models \bigwedge_i \atom_i(s(\vec b))$
since $s$ is a homomorphism of $\Pi$-structures.
Hence $A \models \atombis(s(\vec b))$ since $A \in \Mod_{\Pi + \Sigma}(\mathbb{B} + \mathbb{Q})$.
It follows that $B \models \atombis((e \comp s)(\vec b))$
since $e$ is a homomorphism of $(\Pi + \Sigma)$-structures,
and we are done since $e \comp s$ is the identity function by assumption.
\end{proof}

We now prove Proposition~\ref{prop:stict-monadicity-quantitativetheory:ladj}.
Let $f,g \colon A \to B$ be a parallel pair in $\Mod_{\Pi + \Sigma}(\mathbb{B} + \mathbb{Q})$,
and assume that the pair $U f, U g \colon U A \to U B$
has a split coequalizer in $\Mod_{\Pi}(\mathbb{B})$, say
\[
  e \colon U B \to X
\]
where $X \in \Mod_{\Pi}(\mathbb{B})$, together with $s \colon X \to U B$ and $t \colon U B \to U A$
similarly as in Equation~\eqref{eq:splitcoeq}.
We have to show that $U$ creates the coequalizing diagram $e f = e g$,
i.e.\ that there exists exactly one $C \in \Mod_{\Pi+\Sigma}(\mathbb{B}+\mathbb{Q})$
and $q \colon B \to C$ such that $U C = X$ and $U q = e$,
and moreover that the unique such $q \colon B \to C$ is a coequalizer of
$f$ and $g$ in $\Mod_{\Pi + \Sigma}(\mathbb{B} + \mathbb{Q})$ (\cite[Definition V.1]{maclane98book}).

We first equip $X$ with a $\Sigma$-structure,
using that $e$ is an absolute coequalizer.
Proceeding similarly as in~\cite[Remark 3.4(7)]{Adamek_Rosicky_1994},
we show that there is a unique $(\Pi + \Sigma)$-structure $C$ on $X$
such that $e \colon U B \to X$ lifts to a $(\Pi + \Sigma)$-homomorphism $q \colon B \to C$.
For each function symbol $\sigma \in \Sigma$ of arity $n$,
since $e \comp U f = e \comp U g$ in $\Mod_\Pi(\mathbb{B})$
we have $e \comp U f \comp \sigma^A = e \comp U g \comp \sigma^A$ in $\Set$.
Hence, in $\Set$,
\[
  e \comp \sigma^B \comp (U f \times \cdots \times U f)
  =
  e \comp \sigma^B \comp (U g \times \cdots \times U g)
\]
since $f$ and $g$ are $(\Pi + \Sigma)$-homomorphisms.
It follows that $e \comp \sigma^B$
coequalizes $U f \times \cdots \times U f$
and $U g \times \cdots \times U g$ in $\Set$.
But since $e$ is an absolute coequalizer, it is preserved by
applying the carrier functor to $\Set$ and then the $n$-fold product functor on $\Set$.
Hence $e \times \cdots \times e$
is a coequalizer of 
$U f \times \cdots \times U f$
and $U g \times \cdots \times U g$ in $\Set$.
As a result, there is a unique $\sigma^C$ such that
\begin{equation}
\label{eq:e:hom}
  e \comp \sigma^B
  =
  \sigma^C \comp
  (e \times \cdots \times e)
\end{equation}
This yields an interpretation of each function symbol $\sigma \in \Sigma$,
thus lifting the $\Pi$-structure on $X$ to a $(\Pi + \Sigma)$-structure
$C \in \Str(\Pi+\Sigma)$, with $U C = X$.
By construction of $C$, $e$ lifts to a homomorphism of $(\Pi + \Sigma)$-structures $q \colon B \to C$
with $U q = e$.

This $(\Pi + \Sigma)$-structure $C$ on $X$ is unique
such that $U C = X$ and $U q = e$.
This uses the faithfulness of $U$ and the fact that $e \comp s$ is the identity.
Given any $\Sigma$-structure $C'$ on $X$
and any $(\Pi + \Sigma)$-homomorphism $q' \colon B \to C'$
such that $U C' = X$ and $U q' = e$,
for each $\sigma \in \Sigma$,
we must have, in $\Set$,
\[
  \sigma^{C'} 
  =
  \sigma^{C'} \comp (e \times \cdots \times e) \comp (s \times \cdots \times s)
  =
  e \comp \sigma^B \comp (s \times \cdots \times s)
  =
  \sigma^C \comp (e \times \cdots \times e) \comp (s \times \cdots \times s)
  =
  \sigma^C
\]
while $q' = q$ since $U q' = e = U q$ with $U$ faithful.











Moreover, $C \in \Mod_{\Pi+\Sigma}(\mathbb{B}+\mathbb{Q})$
by Lemma~\ref{lem:horn:splitepi}.
We thus have unique $C \in \Mod_{\Pi+\Sigma}(\mathbb{B}+\mathbb{Q})$ and $q \colon B \to C$
such that $U C = X$ and $U q = e$.

It remains to show that $q \colon B \to C$ is a coequalizer of
$f$ and $g$ in $\Mod_{\Pi + \Sigma}(\mathbb{B} + \mathbb{Q})$.
This follows from easy manipulations, using that
$U q = e$ is a coequalizer in $\Mod_\Pi(\mathbb{B})$
and that $q\colon B \to C$ is a homomorphism in $\Mod_{\Pi + \Sigma}(\mathbb{B} + \mathbb{Q})$.
Let $p \colon B \to K$ in $\Mod_{\Pi + \Sigma}(\mathbb{B} + \mathbb{Q})$
such that $p f = p g$.
We thus have $(U p) \comp (U f) = (U p) \comp (U g)$,
and since $e$ is a coequalizer of $U f$ and $U g$ in $\Mod_\Pi(\mathbb{B})$,
there is a unique $h \colon U C \to U K$ in $\Mod_\Pi(\mathbb{B})$
such that $U p = h e$.

We show that $h$ is a homomorphism of $\Sigma$-structures.
Given $\sigma \in \Sigma$, we have to show that, in $\Set$,
\[
  h \comp \sigma^C
  =
  \sigma^K \comp (h \times \cdots \times h)
\]

\noindent
But by Equation~\eqref{eq:e:hom} above,
and since $p$ is a $\Sigma$-homomorphism,
we have, in $\Set$,
\[
\begin{array}{l l l}
  h \comp \sigma^C
& =
& h \comp \sigma^C \comp (e \times \cdots \times e) \comp (s \times \cdots \times s)
\\

& =
& h \comp e \comp \sigma^B
  \comp (s \times \cdots \times s)
\\

& =
& p \comp \sigma^B
  \comp (s \times \cdots \times s)
\\

& =
& \sigma^K \comp (p \times \cdots \times p)
  \comp (s \times \cdots \times s)
\\

& =
& \sigma^K \comp (h \times \cdots \times h) \comp (e \times \cdots \times e)
  \comp (s \times \cdots \times s)
\\

& =
& \sigma^K \comp (h \times \cdots \times h)
\end{array}
\]

In other words,
$h$ lifts to a $(\Pi+\Sigma)$-homomorphism $k \colon C \to K$ with $U k = h$.
Moreover $k$ is unique such that $k q = p$
since $k' q = p$ implies $U k' = h = U k$,
and thus $k' = k$ since $U$ is faithful.

This concludes the proof of
Proposition~\ref{prop:stict-monadicity-quantitativetheory:ladj}.
\qed\@

\section{Compact Theories of Quantitative Algebras}\label{sec:compact-definition}

The notion of compact theory of quantitative algebras has been introduced in \cite{entics:16876}.
The mathematical setting of \cite{entics:16876} is a special case of that considered in this paper (in Section~\ref{sec:quantitative-algebra-background}), where:
\begin{enumerate}
    \item The base theory $(\Pi, \mathbb{B})$ is fixed to be that of fuzzy relations, as described in Example~\ref{example:base-theory-fuzzy-relations}.
    \item The basic quantitative equations in $\mathbb{Q}$ are required to be \emph{antecedent complete} (see Definition~\ref{def:antecedente-complete} below). As we will discuss in Remark~\ref{remark:AC-not-true-restriction}, this is not truly a restriction, but the notion of \emph{antecedent complete basic quantitative equation} is key for the definition of compactness. 
\end{enumerate}
In what follows, we extend the definition of compact theories from \cite{entics:16876} to the more general setting of Section~\ref{sec:quantitative-algebra-background}. The definition coincides
with that of \cite{entics:16876} when the two restrictions above are adopted.

Recall from Definition~\ref{def:presentation-operation} that, for every $A\in \Mod_\Pi(\mathbb{B})$, we have a presentation $\Pres(A)=(\vec x, \Phi_A )$,
where $\vec x = (x_a)_a$ consists of one variable $x_a$ for each $a\in A$, and $\Phi_A$ is the set of all atomic $\Pi$-formulas that hold in $A$.
\begin{definition}[Antecedent Complete Basic Quantititive Equations]\label{def:antecedente-complete}
A basic quantitative equation 
\[
  \forall \vec{x}. \Big( \bigwedge_{i \in I} \atom_i(\vec{x}) \Big)
  \Rightarrow \atomext(\vec{x})
\]
is called \emph{antecedent complete} (or just ``AC'') if
$(\vec x, \{\atom_i(\vec{x})\}_{i \in I}) = \Pres(A)$,
for some $A\in  \Mod_\Pi(\mathbb{B})$. 
\end{definition}

In other words, in AC basic quantitative equations,
the variables are in one-to-one correspondance with the elements of a model
$A\in \Mod_\Pi(\mathbb{B})$ in the base category, 
and the the formulas in the antecedent are a complete description of the relations
which hold in $A$.
Note that both $\vec{x}$ and $\{\atom_i(\vec{x})\}_{i \in I}$
can be infinite.

\begin{remark}\label{remark:AC-not-true-restriction}
Alternatively, a basic quantitative equation is AC precisely when
\begin{enumerate}
    \item for every $\Pi$-atomic formula $\atom'(\vec{x})$, if $\atom'(\vec{x})$ is a 
    logical consequence of the assumptions $\{ \atom_i(\vec{x})\}_{i \in I}$ under the base theory $\mathbb{B}$, 
then it must be already contained in the assumptions, i.e.:
\[
\mathbb{B} + \{ \atom_i(\vec{x})\}_{i \in I}  \models \atom'(\vec{x}) \qquad \textnormal{implies} \qquad  \atom'(\vec{x}) \in \{ \atom_i(\vec{x})\}_{i \in I} 
\]


\item 
and in $\{ \atom_i(\vec{x})\}_{i \in I}$,
there is no formula $x=x'$ with $x, x'$ distinct variables.
\end{enumerate}
 Consequently, for each $T= \langle (\Pi,\mathbb{B}), (\Sigma,\mathbb{Q})\rangle$, one can always find a  $\mathbb{Q'}$, consisting only of AC basic quantitative equations,
such that $T'= \langle (\Pi,\mathbb{B}), (\Sigma,\mathbb{Q'})\rangle$ is equivalent to $T$, in the sense that $\Mod_{\Pi+\Sigma}(\mathbb{B}+\mathbb{Q})= \Mod_{\Pi+\Sigma}(\mathbb{B}+\mathbb{Q'})$. 
Such $\mathbb{Q'}$ is obtained by ``completing'' (i.e., adding all deducible atomic formulas to) the antecedents of all basic quantitative equations in $\mathbb{Q}$ under the theory $\mathbb{B}$,
and identifying variables $x=x'$ in the quantifications, if $x=x'$ is deducible.
\end{remark}
\begin{remark}
In \cite{entics:16876} and \cite{MioEA24}, since the base theory is fixed to be that of fuzzy relations (i.e., $\Mod_\Pi(\mathbb{B}) = \FRel$), AC basic quantitative equations 
are written as ``$\forall (X,d_X). \atomext(\vec{x})$'', to explicitly indicate that the atomic formulas of the antecedent are exactly those holding true in the fuzzy relation $(X,d_X)\in \FRel$.
\end{remark}

In order to introduce the notion of compact quantitative equational theory, from \cite{entics:16876}, we need the following technical definition.
\begin{definition}[Finitary Restriction of a Theory]\label{def:finitary-restriction}
Given a quantitative theory $T= \langle (\Pi,\mathbb{B}), (\Sigma,\mathbb{Q})\rangle$, we define its \emph{finitary restriction} to be:
\[
T_f= \langle (\Pi,\mathbb{B}_f), (\Sigma,\mathbb{Q}_f)\rangle
\]
where $\mathbb{B}_f \subseteq \mathbb{B}$ and $\mathbb{Q}_f\subseteq \mathbb{Q}$ are the subsets of formulas, in $\mathbb{B}$ and $\mathbb{Q}$ respectively,
that are expressed in ordinary first order logic, i.e.,
that involve only a finite number of variables and formulas in the antecedent.
\end{definition}

\begin{example}
Consider the theory $T= \langle (\Pi,\mathbb{B}), (\Sigma,\mathbb{Q})\rangle$ of quantitative convex algebras over fuzzy relations, from Example~\ref{example:quantitative-convex-algebras}.
In this case $\mathbb{B}_f$ is a proper subset of $\mathbb{B}$, obtained by removing all instances of the \emph{order completeness} axioms (see Example~\ref{example:base-theory-fuzzy-relations}) which are infinitary. 
By contrast, $\mathbb{Q}_f= \mathbb{Q}$, because all basic quantitative equations in $\mathbb{Q}$ are finitary.
\end{example}

We are now ready to formally define compact quantitative theories.
\begin{definition}[Compact Quantitative Theories]\label{compactness_definition}
A theory $T= \langle (\Pi,\mathbb{B}), (\Sigma,\mathbb{Q})\rangle$ is \emph{compact} if, for all AC basic quantitative equations $H$ of the form $\forall \vec{x}. \Big( \bigwedge_{i \in I} \alpha_i(\vec{x}) \Big) \Rightarrow \atomext(\vec{x})$,
the following holds:
\[
  \mathbb{B}+\mathbb{Q}\models H
  \qquad\Longrightarrow\qquad
  \mathbb{B}_f+\mathbb{Q}_f\models H
\]
i.e., if $H$ is a logical consequence of the theory $\mathbb{B}+\mathbb{Q}$
then it must also be a consequence of the weaker theory
$\mathbb{B}_f+\mathbb{Q}_f\subseteq  \mathbb{B}+\mathbb{Q}$
(note that the other direction $(\Longleftarrow)$ is always trivially true).
\end{definition}

\begin{remark}
Note that, in the definition of compact theory, the conservativity requirement is only relative to AC basic quantitative equations $H$. 
It may therefore be possible, for a theory $T$,
to be compact and that the implication
$\mathbb{B}+\mathbb{Q}\models H \not\Rightarrow \mathbb{B}_f+\mathbb{Q}_f\models H$
fails for some non--AC basic quantitative equation $H$.
\end{remark}
\begin{remark}\label{remark:finite-proofs}
By ordinary logical reasoning, the following holds:
\[
   \mathbb{B}_f+\mathbb{Q}_f \models
   \forall \vec{x}. \Big( \bigwedge_{i \in I} \atom_i(\vec{x}) \Big)
   \Rightarrow \atomext(\vec{x})
   \qquad\Longleftrightarrow\qquad
   \mathbb{B}_f+\mathbb{Q}_f + \{ \atom(\vec{x})\}_{i \in I} \models  \atombis(\vec{x})
\]

\noindent
with $\vec{x}$ fresh in $\mathbb{B}_f$ and $\mathbb{Q}_f$, and
where free variables in sequents are interpreted as being universally quantified,
as customary.
Note that while $\vec{x}$ and the index set $I$ can be infinite,
the right-hand side expresses logical entailment between ordinary first order logic
formulas, because all function and relation symbols $\Pi$ and $\Sigma$ have finite arity.
Therefore, by the compactness theorem of first order logic,
the right-hand side holds if and only if there exists a \emph{finite} proof
of the sequent
$\mathbb{B}_f+\mathbb{Q}_f + \{ \alpha_i(\vec{x})\}_{i \in I} \vdash \atomext(\vec{x})$,
necessarily using only a finite subset of the formulas on the left-hand side
of the sequent.
In \cite{entics:16876},
this ``finite proof'' formulation for the compactness of a theory is adopted,
whence the name due to the involvement of the compactness theorem of first order logic.
It is equivalent to the formulation in Definition~\ref{compactness_definition}. 
\end{remark}

\begin{example}
The main result of \cite{entics:16876} is that the quantitative equational
theory of convex algebras (in Example~\ref{example:quantitative-convex-algebras})
is compact. 
\end{example}

\section{Main Result}\label{sec:main-result}

Recall from Definition~\ref{def:finitary-restriction} that,
for any quantitative theory $T= \langle (\Pi,\mathbb{B}), (\Sigma,\mathbb{Q})\rangle$
we have a corresponding subtheory, its finitary restriction
$T_f= \langle (\Pi,\mathbb{B}_f), (\Sigma,\mathbb{Q}_f)\rangle$. 
Since
$\mathbb{B}_f \subseteq \mathbb{B}$
and
$\mathbb{B}_f + \mathbb{Q}_f \subseteq \mathbb{B} + \mathbb{Q}$,
the category $\Mod_\Pi (\mathbb{B})$ is a full subcategory of $\Mod_\Pi (\mathbb{B}_f)$
and
$\Mod_\Pi (\mathbb{B} + \mathbb{Q})$ is a full subcategory of 
$\Mod_\Pi (\mathbb{B}_f + \mathbb{Q}_f)$
(see Theorem~\ref{thm:horn:adj}(\ref{item:horn:adj:forget})).
Denote with $G$ and $I$ the inclusion functors
\[
  G\colon \Mod_\Pi (\mathbb{B}) \hookrightarrow \Mod_\Pi (\mathbb{B}_f)
  \qquad\text{and}\qquad 
  I\colon
  \Mod_{\Pi+\Sigma} (\mathbb{B}+\mathbb{Q})
  \hookrightarrow
   \Mod_{\Pi+\Sigma} (\mathbb{B}_f+\mathbb{Q}_f).
\]

Let 
$\left( T:=U_TF_T \colon \Mod_\Pi(\mathbb{B})\to \Mod_\Pi (\mathbb{B}) , \eta^{T}, \mu^{T}\right)$ 
be the monad arising from the adjunction of
Theorem~\ref{thm:stict-monadicity-quantitativetheory} and also,
as any monad,
from the Eilenberg--Moore adjunction (see Proposition~\ref{prop:key-property-comparison}).
The two are connected by the comparison functor $K_T$ (see Section~\ref{sec:monad-background-comparison}) which, by Theorem~\ref{thm:stict-monadicity-quantitativetheory}, is an isomorphism of categories. Similarly for the monad $T_f$. We thus have the following functors involved:
\[\begin{tikzcd}
	{\Mod_{\Pi+\Sigma}(\mathbb{B}+\mathbb{Q})} & {\Mod_{\Pi}(\mathbb{B)}} && {\Mod_\Pi(\mathbb{B}_f)} & {\Mod_{\Pi+\Sigma}(\mathbb{B}_f+\mathbb{Q}_f)} \\
	\\
	{\EM(T)} &&&& {\EM(T_f)}
	\arrow["{F_T}", shift left=3, tail reversed, no head, from=1-1, to=1-2]
	\arrow["{K_T}", from=1-1, to=3-1]
	\arrow["{U_T}", tail reversed, no head, from=1-2, to=1-1]
	\arrow["G", hook, from=1-2, to=1-4]
	\arrow["{U^{\star}_T}", shift left=3, tail reversed, no head, from=1-2, to=3-1]
	\arrow["{F_{T_f}}", shift left=3, from=1-4, to=1-5]
	\arrow["{F^{\star}_{T_f}}"{pos=0.7}, from=1-4, to=3-5]
	\arrow["{U_{T_{f}}}", from=1-5, to=1-4]
	\arrow["{K_{T_f}}", from=1-5, to=3-5]
	\arrow["{F^\star_T}", tail reversed, no head, from=3-1, to=1-2]
	\arrow["{U^{\star}_{T_f}}"{pos=0.3}, shift left=3, from=3-5, to=1-4]
  \arrow["I", hook, bend left=25, swap, from=1-1, to=1-5]
\end{tikzcd}\]
with the following commutations:
$K_T F_T = F^\star_T$,
$K_{T_f} F_{T_f} = F^\star_{T_f}$,
$U^\star_T K_T = U_T$,
$U^\star_{T_f} K_{T_f} = U_{T_f}$,
by Proposition~\ref{prop:useful-fact-EM},
and $U_{T_f} I  = G U_T$ (i.e., $I$ lifts $G$)
because all four functors are defined as forgetful functors.
The last observation has the following consequence.

\begin{proposition}\label{unique-nat-trans}
There exists a unique natural $\hat{\theta}\colon  F_{T_f} G \to I F_T$
such that $U_{T_f}(\hat{\theta}) \circ \eta^{T_f} G = G(\eta^T)$.
\end{proposition}
\begin{proof}
The natural transformation $G(\eta^{T})$ has type $G\to G(U_T F_T)$ but, since $I$ lifts $G$, its type is equivalently $G\to U_{T_f} I  F_T$. 
Therefore, by the universal property of the adjunction $F_{T_f}\dashv U_{T_f}$,
there exists 
a unique $\hat{\theta}\colon  F_{T_f} G \to I F_T$
satisfying $U_{T_f}(\hat{\theta}) \circ \eta^{T_f}G = G(\eta^T)$.  
\end{proof}

The following is our main result.

\begin{definition}\label{def:compatibility-of-nat-trans}
  We say that a monad morphism $(G,\gamma\colon T_f G\Rightarrow G T)$ (see Definition~\ref{def:monad_morphism}) is \emph{compatible} with the inclusion functor $I$ if the corresponding canonical functor $\Gamma_\gamma\colon\EM(T)\to \EM(T_f)$ (see Proposition~\ref{prop:from_morphism_to_functor})
  coincides with $I$ under the comparison functors, i.e., $I =  K^{-1}_{T_f}\Gamma_\gamma K_T $. 
\end{definition}

\begin{theorem}\label{thm:main-result}
For every quantitative theory $T=\langle (\Pi,\mathbb{B}), (\Sigma,\mathbb{Q})\rangle$, the following are equivalent:
\begin{enumerate}
\item
\label{item:main-result:compact}
The theory $T$ is compact (see Definition~\ref{compactness_definition}).

\item
\label{item:main-result:theta}
The natural transformation $\hat{\theta}\colon  F_{T_f} G \to I F_T$
in Proposition~\ref{unique-nat-trans} is a natural isomorphism. 

\item
\label{item:main-result:monad}
There is a natural isomorphism $\gamma$ yielding a monad restriction $(G,\gamma\colon T_f G\Rightarrow G T)$ such that $\gamma$ (which is essentially unique by Proposition \ref{prop:restriction_is_unique}) is compatible with the inclusion functor $I$. 
\end{enumerate}
\end{theorem}

In particular, from
\(
\text{(\ref{item:main-result:compact})}
\Rightarrow
\text{(\ref{item:main-result:monad})}
\)
in Theorem~\ref{thm:main-result}
we obtain the following corollary.

\begin{corollary}\label{cor:consequence-of-main-theorem}
If the theory $T$ is compact then the monad $T$ is a restriction of $T_f$ along $G$. 
\end{corollary}
\begin{remark}\label{remark:open-problem}
Implication
\(
  \text{(\ref{item:main-result:monad})}
  \Rightarrow
  \text{(\ref{item:main-result:compact})}
\)
in Theorem~\ref{thm:main-result}
does not seem to
readily imply the converse of the implication in Corollary~\ref{cor:consequence-of-main-theorem}.
It is indeed \emph{a priori} conceivable that monad restrictions
$(G,\gamma\colon T_f G\Rightarrow G T)$ may exist and yet none of them
are compatible with $I$
(Proposition~\ref{prop:restriction_is_unique} states that they are all equal
up to isomorphism).
We do not know if the converse of the implication in Corollary~\ref{cor:consequence-of-main-theorem} holds.
\end{remark}

The rest of this paper is devoted to the proof of Theorem~\ref{thm:main-result}. 
We prove direction 
\(
\text{(\ref{item:main-result:compact})}
\Rightarrow
\text{(\ref{item:main-result:theta})}
\)
in Subsection~\ref{proof-direction-1-to-2},
direction
\(
\text{(\ref{item:main-result:theta})}
\Rightarrow
\text{(\ref{item:main-result:compact})}
\)
in Subsection~\ref{proof-direction-2-to-1},
direction
\(
\text{(\ref{item:main-result:monad})}
\Rightarrow
\text{(\ref{item:main-result:theta})}
\)
in Subsection~\ref{proof-direction-3-to-2}
and direction
\(
\text{(\ref{item:main-result:theta})}
\Rightarrow
\text{(\ref{item:main-result:monad})}
\)
in Subsection~\ref{proof-direction-2-to-3}.
Before doing so, we state as propositions the following instances,
to the present situation and notation, of previously established facts.

\begin{proposition}[from Corollary~\ref{cor:restriction_lemma_to_be_used}]
  \label{key_lemma_1}
$T$ is a monad restriction of $T_f$ along $G$ if, and only if, there exists a functor $\Gamma\colon \EM(T)\to \EM(T_f)$ lifting $G$ (i.e., $ U^{\star}_{T_f}\circ \Gamma = G\circ U^\star_{T}$) and such that the canonical natural transformation
 $\widehat{\gamma}^{\Gamma}\colon F^\star_{T_f}\circ G \Rightarrow \Gamma \circ F^\star_T$ (as specified in Proposition~\ref{prop:from_functor_to_morphism}) 
is a natural isomorphism.
\end{proposition}

\begin{proposition}[from
Proposition~\ref{prop:main-fact-presentations}(\ref{item:main-fact-presentations:pres})]
\label{key_lemma_2}
For all $A\in \Mod_{\Pi}(\mathbb{B})$, let $\Pres(A)=(\vec x , \Phi_A)$. Then: 
\begin{enumerate}
\item
\label{item:key_lemma_2:inf}
$\big( F_T(A), (\eta^{T}_A(a) \mid a\in A) \big)$
is presented by
$\Pres(A)$ in $\Mod_{\Pi+\Sigma}(\mathbb{B}+\mathbb{Q})$, and

\item
\label{item:key_lemma_2:fin}
$\big( F_{T_f}(GA), (\eta^{T_f}_{GA}(a) \mid a\in A) \big)$
is presented by
$\Pres(A)$ in $\Mod_{\Pi+\Sigma}(\mathbb{B}_f+\mathbb{Q}_f)$.
\end{enumerate}
\end{proposition}

\begin{proof}
The first item is a direct instantiation of the first item in
Proposition~\ref{prop:main-fact-presentations} by taking
$\Mod_\Sig(\Horn)=\Mod_{\Pi}(\mathbb{B})$
and $\Mod_\Sigbis(\Hornbis)=\Mod_{\Pi+\Sigma}(\mathbb{B}+\mathbb{Q})$.
For the second item, a similar instantiation formally gives: 
\begin{center}
 ``$ \big( F_{T_f}(GA), (\eta^{T_f}_{GA}(a') \mid a'\in GA) \big)$ is presented by $\Pres(GA)$ in $\Mod_{\Pi+\Sigma}(\mathbb{B}_f+\mathbb{Q}_f)$''
\end{center}
but since the inclusion $G$ is a forgetful functor (i.e., $A$ and $GA$ are the same $\Pi$-structure), the notation $a'\in GA$ and $a\in A$, and similarly $\Pres(A)$ and $\Pres(GA)$, coincide.
Hence the statement~(\ref{item:key_lemma_2:fin}).
\end{proof}

\begin{proposition}[from
Proposition~\ref{prop:main-fact-presentations}(\ref{item:main-fact-presentations:bridge})]
\label{key_lemma_3}
For every AC quantitative equation (see Definition~\ref{def:antecedente-complete})
$\forall \vec{x}. \Big( \bigwedge_{i \in I} \atom_i(\vec{x}) \Big) \Rightarrow \atomext(\vec{x})$, where $(\vec x, \{ \atom_i(\vec{x})\}_{i \in I}) = \Pres(A)$ for some $A\in\Mod_{\Pi}(\mathbb{B})$, 
it holds that:
\begin{enumerate}
\item $
  \big( F_T(A), (\eta^{T}_A(a))_{a\in A} \big) \models \atomext(\vec{x})
  \qquad \Longleftrightarrow \qquad
  \mathbb{B}+\mathbb{Q} \models \forall \vec{x}. \Big( \bigwedge_{i \in I} \atom_i(\vec{x}) \Big) \Rightarrow \atomext(\vec{x})$, and
\item $
   \big( F_{T_f}(GA), (\eta^{T_f}_{GA}(a))_{a\in A} \big) \models \atomext(\vec{x})
  \qquad \Longleftrightarrow \qquad
  \mathbb{B}_f+\mathbb{Q}_f \models \forall \vec{x}. \Big( \bigwedge_{i \in I} \atom_i(\vec{x}) \Big) \Rightarrow \atomext(\vec{x})$.
  \end{enumerate}
\end{proposition}
\begin{proof}
On top of the same argument used in the proof of Proposition~\ref{key_lemma_2},
it is sufficient to observe that quantifying over all $A\in \Mod_{\Pi}(\mathbb{B})$
and atomic formulas $\atomext(\vec{x})$,
as in the statement of
Proposition~\ref{prop:main-fact-presentations}(\ref{item:main-fact-presentations:bridge}),
is equivalent to quantifying over all AC basic quantitative equations
(see Definition~\ref{def:antecedente-complete}).
\end{proof}


\subsection{Proof of Direction
\(
  \text{(\ref{item:main-result:compact})}
  \Rightarrow
  \text{(\ref{item:main-result:theta})}
\)
in Theorem~\ref{thm:main-result}}
\label{proof-direction-1-to-2}


We need to prove that, for each $A\in\Mod_\Pi(\mathbb{B})$, the component  $\widehat{\theta}_A$ of the natural transformation $\widehat{\theta}$ of Proposition~\ref{unique-nat-trans} is an
isomorphism.
Denote with $\theta = U_{T_f}(\widehat{\theta})$ the underlying morphism which satisfies
$\theta_A \circ \eta^{T_f}_{GA}  = G(\eta^{T}_A)$.
We are going to prove that $\theta_A$ is an isomorphism using
Lemma~\ref{lem:pres:unique-iso}. This, in turn, will imply that 
$\widehat{\theta}_A$ is an isomorphism, because $U_{T_f}$ reflects isomorphisms
(Theorem~\ref{thm:stict-monadicity-quantitativetheory}(\ref{item:reflection-of-iso})).

To invoke Lemma~\ref{lem:pres:unique-iso}, we need the following four observations.
The assumption that the theory $T$ is compact is used in Observation 4.

\emph{Observation 1.} As noted in Proposition~\ref{unique-nat-trans},
the natural transformation $G(\eta^T_A)$
has, equivalently, type $GA \to  U_{T_f} (IF_T (A))$. This means that $G(\eta^T_A)(a)$ is a point of the
model
$I F_T(A) \in \Mod_{\Pi+\Sigma}(\mathbb{B}_f+\mathbb{Q}_f)$,
for each $a\in GA$.
Or, equivalently, for each $a\in A$, because  $G$ is a forgetful functor (i.e., $A$ and $GA$ are the same structures). 
Therefore we have the following expanded $(\Pi+\Sigma)$-structure:
\[
\left( I F_T(A), (G(\eta^T_A)(a) \mid a\in A) \right)
\]

\emph{Observation 2.} 
The property $\theta_A \circ \eta^{T_f}_{GA}  = G(\eta^{T}_A)$ means exactly that
${\theta}_A$ is a homomorphism
of expanded structures, i.e., it preserves all the parameters:
\[
  {\theta}_A\colon
  \big( U_{T_f} F_{T_f}(GA)  , (\eta^{T_f}_{GA}(a)\mid a\in A) \big)
  ~\to~
  \big( U_{T_f} I F_T(A) , (G(\eta^T_A)(a) \mid a\in A) \big)
\]

\emph{Observation 3.}
The parameters  $(G(\eta^T_A)(a) \mid a\in A)$ generates $I F_T(A)$.
This is because, by Proposition~\ref{key_lemma_2}(\ref{item:key_lemma_2:inf}),
the $(\Pi+\Sigma)$-structure
$(F_T(A), (\eta^T_A(a)\mid a\in A))$ is presented by $\Pres(A)$ in $\Mod_{\Pi+\Sigma}(\mathbb{B}+\mathbb{Q})$ and, by Definition~\ref{def:horn:pres}, this implies that
$(\eta^T_A(a)\mid a\in A)$ generate $F_T(A)$. But since $I$ and $G$ are defined as forgetful functors (i.e., $IF_T(A)$ and $F_T(A)$ are the same structures and $\eta^T_A(a)$ and $G(\eta^T_A)(a)$
are the same parameters) we deduce that $(G(\eta^T_A)(a) \mid a\in A)$ generate $I F_T(A)$.\\

\emph{Observation 4.} By Definition~\ref{compactness_definition}, the assumption that $T$ is compact means that for all AC basic quantitative equations $H$ of the form  
$\forall \vec{x}. \Big( \bigwedge_{i \in I} \atom_i(\vec{x}) \Big) \Rightarrow \atomext(\vec{x})$, where $(\vec x, \{ \atom_i(\vec{x})\}_{i \in I}) = \Pres(B)$
for some $B\in\Mod_{\Pi}(\mathbb{B})$, it holds that $\mathbb{B}_f+\mathbb{Q}_f\models H  \Rightarrow \mathbb{B} +\mathbb{Q} \models H$.
In fact, since the reverse implication trivially holds (because $\mathbb{B}_f+\mathbb{Q}_f\subseteq \mathbb{B}+\mathbb{Q})$, we have:
\[
  \mathbb{B}_f+\mathbb{Q}_f\models H
  \qquad\Longleftrightarrow\qquad
   \mathbb{B} +\mathbb{Q} \models H.
\]

In particular, by fixing $B$ to be $A$, and using Proposition~\ref{key_lemma_3}, the assumption implies that
\[
 \big( F_{T_f}(GA), (\eta^{T_f}_{GA}(a))_{a\in A} \big) \models \atomext(\vec{x})
 \qquad
 \Longleftrightarrow 
\qquad 
\big( F_T(A), (\eta^{T}_A(a))_{a\in A} \big) \models \atomext(\vec{x})
\]
holds for all atomic $\atomext(\vec{x})$ in the signature $\Pi+\Sigma$. Again noting that $G$ and $I$ are forgetful functors, we can rewrite the left-hand side and obtain, for all atomic $\atomext(\vec{x})$ in the signature $\Pi+\Sigma$:
\[
  \big( IF_T(A), (G(\eta^{T}_A)(a))_{a\in A} \big) \models \atomext(\vec{x})
  \qquad\Longleftrightarrow\qquad
  \big( F_{T_f}(GA), (\eta^{T_f}_{GA}(a))_{a\in A} \big) \models \atomext(\vec{x}).
 \]

We can now conclude that $\theta_A$ is an isomorphism by instantiating the statement of Lemma~\ref{lem:pres:unique-iso} as follows:
 $\Sig = \Pi$, $\Sigbis =\Pi+\Sigma$, $\Horn=\mathbb{B}_f$, $\Hornbis= \mathbb{B}_f+  \mathbb{Q}_f$, $F=F_{T_f}$, $U=U_{T_f}$, $B= IF_T(A)$ and 
 the list of parameters $(b_a \mid a\in A)$ as  $(G(\eta^{T}_A)(a) \mid a\in A)$. The required assumptions of Lemma~\ref{lem:pres:unique-iso} are
 satisfied by Observations 1,3 and 4. Therefore $\theta_A$ which, by Observation 2,
is a homomorphism of expanded structures, is the only such one and is an isomorphism.

\subsection{Proof of Direction
\(
  \text{(\ref{item:main-result:theta})}
  \Rightarrow
  \text{(\ref{item:main-result:compact})}
\)
in Theorem~\ref{thm:main-result}}
\label{proof-direction-2-to-1}

To prove that the theory $T$ is compact (see Definition~\ref{compactness_definition})
we need to show that for every AC basic quantitative equation $H$ of the form
\(
  \forall \vec{x}. \Big( \bigwedge_{i \in I} \atom_i(\vec{x}) \Big)
  \Rightarrow \atomext(\vec{x})
\),
it holds that
\[
  \mathbb{B}_f + \mathbb{Q}_f\models H
  \qquad\Longrightarrow\qquad
  \mathbb{B} + \mathbb{Q}\models H.
\tag{$\dagger$}
\]
From Definition~\ref{def:antecedente-complete} of AC basic quantitative equations, the variables $\vec{x}$ and formulas $\atom_i(\vec{x})$ come from the presentation of 
some $A\in\Mod_\Pi(\mathbb{B})$, i.e., $\Pres(A)=( \vec{x} , \Phi_A )$, and $\Phi_A =  \{ \atom_i(\vec{x})\}_{i \in I}$. 

Hence, using Proposition~\ref{key_lemma_3}, we can rewrite ($\dagger$) as:
\[
  \big( F_T(A), (\eta^{T}_A(a))_{a\in A} \big) \models \atomext(\vec{x})
  \qquad\Longrightarrow\qquad
  \big( F_{T_f}(GA), (\eta^{T_f}_{GA}(a))_{a\in A} \big) \models \atomext(\vec{x}).
\]
Recall that $I$ and $G$ are defined as forgetful functors, so that $IF_T(A)$ and $F_T(A)$ are the same $(\Pi+\Sigma)$--structure and $G(\eta^{T}_{A})(a)$ and  $\eta^{T}_A(a)$ the same parameters.
Hence we can rewrite the above as: 
\[
  \big( IF_T(A),(G(\eta^{T}_{A})(a))_{a\in A} \big) \models \atomext(\vec{x})
  \qquad\Longrightarrow\qquad
  \big( F_{T_f}(GA), (\eta^{T_f}_{GA}(a))_{a\in A} \big) \models \atomext(\vec{x}).
\]
By assumption~(\ref{item:main-result:theta}),
we have a natural isomorphism
$\hat{\theta}\colon F_{T_f}G  \Rightarrow IF_T$
satisfying $U_{T_f}(\hat{\theta})\circ \eta^{T_f}G = G(\eta^{T})$.
This property means, equivalently, that the component
$\hat{\theta}_A\colon F_{T_f}(GA)\to IF_T(A)$
preserves the parameters and therefore it is an isomorphism of expanded structures: 
\[
  \hat{\theta}_A\colon
  \big( F_{T_f}(GA), (\eta^{T_f}_{GA}(a))_{a\in A} \big)
  \stackrel{\cong}{\longrightarrow}
  \big( IF_T(A),(G(\eta^{T}_{A})(a))_{a\in A} \big) 
\]
Hence ($\dagger$) holds. $\qed$

\subsection{Proof of Direction
$\text{(\ref{item:main-result:monad})} \Rightarrow \text{(\ref{item:main-result:theta})}$
in Theorem~\ref{thm:main-result}}
\label{proof-direction-3-to-2}

We first establish the following lemma, also used in the proof of direction
\(
  \text{(\ref{item:main-result:theta})}
  \Rightarrow
  \text{(\ref{item:main-result:monad})}
\)
in Subsection~\ref{proof-direction-2-to-3}.

\begin{lemma}\label{lemma-simplification-proofs-2-and-3}
Let $\gamma\colon T_{f}G\Rightarrow G T$ be a monad morphism
(not necessarily an isomorphism) compatible with~$I$.
Let $\hat{\theta}\colon  F_{T_f} G \to I F_T$ be the (uniquely determined)
natural transformation from Proposition~\ref{unique-nat-trans} satisfying $U_{T_f}(\hat{\theta}) \circ \eta^{T_f}G = G(\eta^T)$. Then it holds that:
$U_{T_f}(\hat{\theta}) = \gamma$.
\end{lemma}
\begin{proof}
The \emph{unit compatibility} property of the monad morphism $\gamma$ (see Definition~\ref{def:monad_morphism}) states that:
\[
\gamma\circ (\eta^{T_{f}} G) = G(\eta^{T})  
\tag{Unit compatibility of $\gamma$}
\] 
From $\gamma$, by application of Proposition~\ref{prop:from_morphism_to_functor} we obtain a canonical functor $\Gamma_\gamma\colon \EM(T)\to \EM(T_f)$ lifting $G$ (i.e., $GU^\star_{T} = U^\star_{T_f}\Gamma_\gamma$).
And in turn, combining Proposition~\ref{prop:from_functor_to_morphism} and Proposition~\ref{prop:construction-inversion}, from $\Gamma_\gamma$ we obtain a canonical natural transformation $\widehat{\gamma}\colon F^{\star}_{T_f}G\Rightarrow  \Gamma_\gamma F^\star_{T}$ such that $U^\star_{T_{f}}(\widehat{\gamma})=\gamma$.

The assumption that $\gamma$ is compatible with $I$ is, by Definition \ref{def:compatibility-of-nat-trans}, that $K^{-1}_{T_f}\Gamma_\gamma K_T = I$.

Let $\widehat{\tau}= K^{-1}_{T_f}\widehat{\gamma}\colon F_{T_f}G\Rightarrow  I F_{T}$ be the natural transformation corresponding to $\widehat{\gamma}$ via the comparison functor.
Note that, by definition of $K_{T_f}$ on morphisms (see Section~\ref{sec:monad-background-comparison}) we have that $U_{T_f}(\widehat{\tau}) = U^\star_{T_f}(\widehat{\gamma})=\gamma$.
This implies that the equation $U_{T_f}(\widehat{\tau}) \circ \eta^{T_f}G = G(\eta^T)$ holds. 
By Proposition~\ref{unique-nat-trans}, there exists only one such natural transformation.
Hence $\widehat{\tau}=\widehat{\theta}$.
And therefore $\gamma = U_{T_f}(\widehat{\tau}) = U_{T_f}(\widehat{\theta})$.
\end{proof}

We now return to the proof of Direction
$\text{(\ref{item:main-result:monad})} \Rightarrow \text{(\ref{item:main-result:theta})}$
in Theorem~\ref{thm:main-result}. From the assumption, we have a monad morphism
$\gamma\colon T_{f} G \Rightarrow G T$
which is a natural isomorphism and such that $\gamma$ is compatible with $I$.
By Lemma~\ref{lemma-simplification-proofs-2-and-3},
$\gamma$ is the underlying morphism of $\widehat{\theta}$.
Reasoning similarly as in the beginning of~\S\ref{proof-direction-1-to-2},
we get that $\widehat{\theta}$ is a natural isomorphism.

\subsection{Proof of Direction
$\text{(\ref{item:main-result:theta})}\Rightarrow\text{(\ref{item:main-result:monad})}$
in Theorem~\ref{thm:main-result}}
\label{proof-direction-2-to-3}

We need to show that there is a monad restriction $(G,\gamma)$,
i.e.\ a monad morphism that is a natural isomorphism,
such that $\gamma$ is compatible with $I$. 

First, observe that it suffices to show that a monad morphism $\gamma$
(not necessarily an isomorphism) such that $\gamma$ is compatible with $I$ exists. 
Indeed, if such $\gamma$ exists, by Lemma~\ref{lemma-simplification-proofs-2-and-3}
it is the underlying morphism of~$\widehat{\theta}$.
Assumption~(\ref{item:main-result:theta})
is that $\widehat{\theta}$ is a natural isomorphism and this implies that $\gamma$ is also a natural isomorphism and therefore $\gamma$ is a monad restriction.

To construct the monad morphism $\gamma$ compatible with $I$ we proceed as follows.
Define the functor $\Gamma\colon \EM(T)\to EM(T_f)$ as the functor $\Gamma =  K_{T_f}\circ I\circ (K_T)^{-1}$, corresponding to $I$ via the comparison functors.
Since $I$ is a lifting of $G$ (i.e., $G U_T = U_{T_f}I$), using Proposition~\ref{prop:useful-fact-EM} we derive that also $\Gamma$ is a lifting of $G$  (i.e., $G U^\star_T = U^\star_{T_f}\Gamma$).
Therefore from $\Gamma$, using Proposition~\ref{prop:from_functor_to_morphism} we obtain a canonical monad morphism $\gamma^{\Gamma}$. 
Define $\gamma$ to be $\gamma^{\Gamma}$. Proposition~\ref{prop:construction-inversion} says that the canonical functor (from Proposition~\ref{prop:from_morphism_to_functor}) 
$\Gamma_{\gamma}$ induced by $\gamma$ is exactly the functor $\Gamma$ we started with. Therefore $\gamma$ is compatible with $I$, 
because $\Gamma_{\gamma} = \Gamma = K_{T_f}\circ I\circ (K_T)^{-1}$. $\qed$

\section{Conclusions and Future Work}\label{sec:conclusions}

In this work,
we have taken the purely logical definition of compact quantitative equational theory
of~\cite{entics:16876}, extended to the more general 
setting of quantitative algebra described in Section~\ref{sec:quantitative-algebra-background}, and provided a categorical characterisation formulated as Theorem~\ref{thm:main-result}.

As mentioned in Remark~\ref{remark:open-problem},
we do not know if the reverse of the implication
in Corollary~\ref{cor:consequence-of-main-theorem} holds.
We leave this as a technical open question for future investigation.



Another interesting direction of future work is towards generalisations. Our notion of compactness (Definition~\ref{compactness_definition}),
expressed as the logical implication 
($
 \mathbb{B}+\mathbb{Q}\models H \Rightarrow \mathbb{B}_f+\mathbb{Q}_f\models H  
$), is based on comparing the theory of $\mathbb{B}+\mathbb{Q}$ to its finitary restriction $\mathbb{B}_f+\mathbb{Q}_f$. A more general approach can, potentially, be based on the notion of
\emph{interpretation} of a theory in another theory
(cf e.g.~\cite{kamsma23aml}).
If the quantitative equational theory  $T= \langle (\Pi,\mathbb{B}), (\Sigma,\mathbb{Q})\rangle$ 
can be interpreted in a finitary (i.e., only having finite formulas) $T' =  \langle (\Pi',\mathbb{B}'), (\Sigma',\mathbb{Q}')\rangle$,
 then we might define 
$\mathbb{B}+\mathbb{Q}$ to be compact \emph{via the interpretation} if 
($
\mathbb{B}+\mathbb{Q}\models H \Rightarrow \mathbb{B}'+\mathbb{Q}' \models \iota(H)
$) where $\iota(H)$ is the interpretation of the basic quantitative equation $H$.
This formulation adds more flexibility, still retaining the key 
fact that validity of judgements $\mathbb{B}'+\mathbb{Q}' \models \iota(H)$
is witnessed by finite proofs due to the compactness theorem of first order logic
(see Remark~\ref{remark:finite-proofs}).
This generalisation might be developed in the context of Morita equivalence
(in the sense of~\cite[D1.4]{johnstone02book}).

In another orthogonal direction, one focusing on the general connection between logic and monad restrictions rather than on finite proofs, 
one could relax the requirement and instead compare $\mathbb{B}+\mathbb{Q}$ to its $\kappa$-ary restriction $\mathbb{B}_\kappa+\mathbb{Q}_\kappa$, defined as expected for any regular cardinal $\kappa$.

\newpage

\bibliographystyle{entics}
\bibliography{biblio}

\begin{thebibliography}{10}
\providecommand{\url}[1]{\texttt{#1}}
\providecommand{\urlprefix}{ }
\providecommand{\eprint}[2][]{\url{#2}}

\bibitem{adamek22}
Ad\'amek, J., \emph{Varieties of quantitative algebras and their monads}, in:
  \emph{37th Annual ACM/IEEE Symposium on Logic in Computer Science (LICS'22)},
  ACM (2022).
\newline\urlprefix\url{https://doi.org/10.1145/3531130.3532405}

\bibitem{adv23}
Ad{\'{a}}mek, J., M.~Dost{\'{a}}l and J.~Velebil, \emph{Strongly finitary
  monads for varieties of quantitative algebras}, in: \emph{10th Conference on
  Algebra and Coalgebra in Computer Science (CALCO'23)}, volume 270 of
  \emph{LIPIcs}, pages 10:1--10:14, Schloss Dagstuhl - Leibniz-Zentrum
  f{\"{u}}r Informatik (2023).
\newline\urlprefix\url{https://doi.org/10.4230/LIPICS.CALCO.2023.10}

\bibitem{Adamek_Rosicky_1994}
Adamek, J. and J.~Rosicky, \emph{Locally Presentable and Accessible
  Categories}, London Mathematical Society Lecture Note Series, Cambridge
  University Press (1994).
\newline\urlprefix\url{https://doi.org/10.1017/CBO9780511600579}

\bibitem{DBLP:journals/lmcs/BacciBLM18}
Bacci, G., G.~Bacci, K.~G. Larsen and R.~Mardare, \emph{A complete quantitative
  deduction system for the bisimilarity distance on markov chains}, Log.
  Methods Comput. Sci. \textbf{14} (2018).
\newline\urlprefix\url{https://doi.org/10.23638/LMCS-14(4:15)2018}

\bibitem{bmpp18}
Bacci, G., R.~Mardare, P.~Panangaden and G.~Plotkin, \emph{An algebraic theory
  of markov processes}, in: \emph{33rd Annual ACM/IEEE Symposium on Logic in
  Computer Science (LICS'18)}, pages 679--688, ACM (2018).
\newline\urlprefix\url{https://doi.org/10.1145/3209108.3209177}

\bibitem{DBLP:journals/corr/abs-2302-01224}
Bacci, G., R.~Mardare, P.~Panangaden and G.~D. Plotkin, \emph{Propositional
  logics for the lawvere quantale}, in: M.~Kerjean and P.~B. Levy, editors,
  \emph{Proceedings of the 39th Conference on the Mathematical Foundations of
  Programming Semantics, {MFPS} XXXIX, Indiana University, Bloomington, IN,
  USA, June 21-23, 2023}, volume~3 of \emph{{EPTICS}}, EpiSciences (2023).
\newline\urlprefix\url{https://doi.org/10.46298/ENTICS.12292}

\bibitem{DBLP:journals/lmcs/BacciMPP24}
Bacci, G., R.~Mardare, P.~Panangaden and G.~D. Plotkin, \emph{Sum and tensor of
  quantitative effects}, Log. Methods Comput. Sci. \textbf{20} (2024).
\newline\urlprefix\url{https://doi.org/10.46298/LMCS-20(4:9)2024}

\bibitem{DBLP:journals/jacm/BacciMPP26}
Bacci, G., R.~Mardare, P.~Panangaden and G.~D. Plotkin, \emph{Quantitative
  equational logic}, J. {ACM} \textbf{73}, pages 26:1--26:52 (2026).
\newline\urlprefix\url{https://doi.org/10.1145/3818603}

\bibitem{BarrWells2005}
Barr, M. and C.~Wells, \emph{Toposes, triples and theories}, Reprints in Theory
  and Applications of Categories pages 1--288 (2005). Originally published by
  Springer-Verlag, New York, 1985.
\newline\urlprefix\url{http://www.tac.mta.ca/tac/reprints/articles/12/tr12.pdf}

\bibitem{DAngeloEA24}
D'Angelo, K., S.~Gurke, J.~M. Kirss, B.~K{\"{o}}nig, M.~Najafi, W.~Rozowski and
  P.~Wild, \emph{Behavioural metrics: Compositionality of the kantorovich
  lifting and an application to up-to techniques}, in: R.~Majumdar and
  A.~Silva, editors, \emph{35th International Conference on Concurrency Theory,
  {CONCUR} 2024, September 9-13, 2024, Calgary, Canada}, volume 311 of
  \emph{LIPIcs}, pages 20:1--20:19, Schloss Dagstuhl - Leibniz-Zentrum
  f{\"{u}}r Informatik (2024).
\newline\urlprefix\url{https://doi.org/10.4230/LIPICS.CONCUR.2024.20}

\bibitem{FordEA21}
Ford, C., S.~Milius and L.~Schr{\"{o}}der, \emph{Monads on categories of
  relational structures}, in: F.~Gadducci and A.~Silva, editors, \emph{9th
  Conference on Algebra and Coalgebra in Computer Science, {CALCO} 2021, August
  31 to September 3, 2021, Salzburg, Austria}, volume 211 of \emph{LIPIcs},
  pages 14:1--14:17, Schloss Dagstuhl - Leibniz-Zentrum f{\"{u}}r Informatik
  (2021).
\newline\urlprefix\url{https://doi.org/10.4230/LIPICS.CALCO.2021.14}

\bibitem{ForsterEA25}
Forster, J., L.~Schr{\"{o}}der, P.~Wild, H.~Beohar, S.~Gurke, B.~K{\"{o}}nig
  and K.~Messing, \emph{Quantitative graded semantics and spectra of
  behavioural metrics}, in: J.~Endrullis and S.~Schmitz, editors, \emph{33rd
  {EACSL} Annual Conference on Computer Science Logic, {CSL} 2025, February
  10-14, 2025, Amsterdam, Netherlands}, volume 326 of \emph{LIPIcs}, pages
  33:1--33:21, Schloss Dagstuhl - Leibniz-Zentrum f{\"{u}}r Informatik (2025).
\newline\urlprefix\url{https://doi.org/10.4230/LIPICS.CSL.2025.33}

\bibitem{Gavazzo18}
Gavazzo, F., \emph{Quantitative behavioural reasoning for higher-order
  effectful programs: Applicative distances}, in: \emph{Logic in Computer
  Science, {LICS} 2018}, pages 452--461, {ACM} (2018).
\newline\urlprefix\url{https://doi.org/10.1145/3209108.3209149}

\bibitem{DBLP:journals/pacmpl/GavazzoF23}
Gavazzo, F. and C.~D. Florio, \emph{Elements of quantitative rewriting}, Proc.
  {ACM} Program. Lang. \textbf{7}, pages 1832--1863 (2023).
\newline\urlprefix\url{https://doi.org/10.1145/3571256}

\bibitem{GoncharovEA23}
Goncharov, S., D.~Hofmann, P.~Nora, L.~Schr{\"{o}}der and P.~Wild,
  \emph{Kantorovich functors and characteristic logics for behavioural
  distances}, in: O.~Kupferman and P.~Sobocinski, editors, \emph{Foundations of
  Software Science and Computation Structures -- 26th International Conference,
  FoSSaCS 2023, Held as Part of the European Joint Conferences on Theory and
  Practice of Software, {ETAPS} 2023, Paris, France, April 22-27, 2023,
  Proceedings}, volume 13992 of \emph{LNCS}, pages 46--67, Springer (2023).
\newline\urlprefix\url{https://doi.org/10.1007/978-3-031-30829-1\_3}

\bibitem{hodges93book}
Hodges, W., \emph{Model Theory}, volume~42 of \emph{Encyclopedia of Mathematics
  and its Applications}, Cambridge University Press (1993).
\newline\urlprefix\url{https://doi.org/10.1017/CBO9780511551574}

\bibitem{johnstone02book}
Johnstone, P., \emph{{Sketches of an Elephant: A Topos Theory Compendium}},
  Oxford Logic Guides, Clarendon Press (2002).
\newline\urlprefix\url{https://doi.org/https://doi.org/10.1093/oso/9780198515982.001.0001}

\bibitem{JurkaEA24}
Jurka, J., S.~Milius and H.~Urbat, \emph{Algebraic reasoning over relational
  structures}, in: V.~De~Paiva and A.~Simpson, editors, \emph{Proc.~40th
  Conference on Mathematical Foundations of Programming Semantics (MFPS)},
  volume~4 of \emph{ENTICS}, pages 13:1--13:20 (2024).
\newline\urlprefix\url{https://doi.org/10.48550/ARXIV.2401.08445}

\bibitem{kamsma23aml}
Kamsma, M., \emph{Type space functors and interpretations in positive logic},
  Arch. Math. Logic. \textbf{62}, pages 1--28 (2023).
\newline\urlprefix\url{https://doi.org/https://doi.org/10.1007/s00153-022-00825-7}

\bibitem{DBLP:conf/fscd/LagoHLP22}
Lago, U.~D., F.~Honsell, M.~Lenisa and P.~Pistone, \emph{On quantitative
  algebraic higher-order theories}, in: A.~P. Felty, editor, \emph{7th
  International Conference on Formal Structures for Computation and Deduction,
  {FSCD} 2022, August 2-5, 2022, Haifa, Israel}, volume 228 of \emph{LIPIcs},
  pages 4:1--4:18, Schloss Dagstuhl -- Leibniz-Zentrum f{\"{u}}r Informatik
  (2022).
\newline\urlprefix\url{https://doi.org/10.4230/LIPICS.FSCD.2022.4}

\bibitem{maclane98book}
Mac~Lane, S., \emph{{Categories for the Working Mathematician}}, Springer, 2nd
  edition (1998).
\newline\urlprefix\url{https://doi.org/https://doi.org/10.1007/978-1-4757-4721-8}

\bibitem{DBLP:conf/lics/MardarePP16}
Mardare, R., P.~Panangaden and G.~D. Plotkin, \emph{Quantitative algebraic
  reasoning}, in: M.~Grohe, E.~Koskinen and N.~Shankar, editors,
  \emph{Proceedings of the 31st Annual {ACM/IEEE} Symposium on Logic in
  Computer Science, {LICS} '16, New York, NY, USA, July 5-8, 2016}, pages
  700--709, {ACM} (2016).
\newline\urlprefix\url{https://doi.org/10.1145/2933575.2934518}

\bibitem{DBLP:conf/lics/MardarePP17}
Mardare, R., P.~Panangaden and G.~D. Plotkin, \emph{On the axiomatizability of
  quantitative algebras}, in: \emph{32nd Annual {ACM/IEEE} Symposium on Logic
  in Computer Science, {LICS} 2017, Reykjavik, Iceland, June 20-23, 2017},
  pages 1--12, {IEEE} Computer Society (2017).
\newline\urlprefix\url{https://doi.org/10.1109/LICS.2017.8005102}

\bibitem{MiliusUrbat19}
Milius, S. and H.~Urbat, \emph{Equational axiomatization of algebras with
  structure}, in: M.~Boja\'nczyk and A.~Simpson, editors,
  \emph{Proc.~Foundations of Software Science and Computation Structures
  (FoSSaCS)}, volume 11425 of \emph{Lecture Notes Comput.~Sci. (ARCoSS)}, pages
  400--417, Springer (2019).
\newline\urlprefix\url{https://doi.org/10.1007/978-3-030-17127-8\_23}

\bibitem{entics:16876}
Mio, M., \emph{Compact quantitative theories of convex algebras}, Electronic
  Notes in Theoretical Informatics and Computer Science \textbf{Volume 5 -
  Proceedings of MFPS XLI}, 15 (2025), ISSN 2969-2431.
\newline\urlprefix\url{https://doi.org/10.46298/entics.16876}

\bibitem{DBLP:conf/lics/MioSV21}
Mio, M., R.~Sarkis and V.~Vignudelli, \emph{Combining nondeterminism,
  probability, and termination: Equational and metric reasoning}, in:
  \emph{36th Annual {ACM/IEEE} Symposium on Logic in Computer Science, {LICS}
  2021, Rome, Italy, June 29 - July 2, 2021}, pages 1--14, {IEEE} (2021).
\newline\urlprefix\url{https://doi.org/10.1109/LICS52264.2021.9470717}

\bibitem{MioEA22}
Mio, M., R.~Sarkis and V.~Vignudelli, \emph{Beyond nonexpansive operations in
  quantitative algebraic reasoning}, in: \emph{Proc.~37th {A}nnual {ACM}/{IEEE}
  {S}ymposium on {L}ogic in {C}omputer {S}cience (LICS'22)}, pages 52:1--52:13,
  ACM (2022).
\newline\urlprefix\url{https://doi.org/10.1145/3531130.3533366}

\bibitem{MioEA24}
Mio, M., R.~Sarkis and V.~Vignudelli, \emph{Universal quantitative algebra for
  fuzzy relations and generalised metric spaces}, Log.~Methods Comput.~Sci.
  \textbf{20}, pages 19:1--19:56 (2024).
\newline\urlprefix\url{https://doi.org/10.46298/LMCS-20(4:19)2024}

\bibitem{DBLP:conf/concur/MioV20}
Mio, M. and V.~Vignudelli, \emph{Monads and quantitative equational theories
  for nondeterminism and probability}, in: I.~Konnov and L.~Kov{\'{a}}cs,
  editors, \emph{31st International Conference on Concurrency Theory, {CONCUR}
  2020, September 1-4, 2020, Vienna, Austria (Virtual Conference)}, volume 171
  of \emph{LIPIcs}, pages 28:1--28:18, Schloss Dagstuhl - Leibniz-Zentrum
  f{\"{u}}r Informatik (2020).
\newline\urlprefix\url{https://doi.org/10.4230/LIPICS.CONCUR.2020.28}

\bibitem{Sarkis2024}
Sarkis, R., \emph{Lifting Algebraic Reasoning to Generalized Metric Spaces},
  Phd thesis, ENS de Lyon (2024). Available at
  \url{https://doi.org/10.5281/zenodo.14001076}.

\bibitem{STREET1972149}
Street, R., \emph{The formal theory of monads}, Journal of Pure and Applied
  Algebra \textbf{2}, pages 149--168 (1972), ISSN 0022-4049.
\newline\urlprefix\url{https://doi.org/https://doi.org/10.1016/0022-4049(72)90019-9}

\bibitem{WildS20}
Wild, P. and L.~Schr{\"{o}}der, \emph{Characteristic {L}ogics for {B}ehavioural
  {M}etrics via {F}uzzy {L}ax {E}xtensions}, in: I.~Konnov and L.~Kov{\'{a}}cs,
  editors, \emph{31st International Conference on Concurrency Theory, {CONCUR}
  2020, September 1-4, 2020, Vienna, Austria (Virtual Conference)}, volume 171
  of \emph{LIPIcs}, pages 27:1--27:23, Schloss Dagstuhl - Leibniz-Zentrum
  f{\"{u}}r Informatik (2020). \eprint{2007.01033}.
\newline\urlprefix\url{https://doi.org/10.4230/LIPIcs.CONCUR.2020.27}

\bibitem{WildS21}
Wild, P. and L.~Schr{\"{o}}der, \emph{A quantified coalgebraic van benthem
  theorem}, in: S.~Kiefer and C.~Tasson, editors, \emph{Foundations of Software
  Science and Computation Structures - 24th International Conference, {FOSSACS}
  2021, Held as Part of the European Joint Conferences on Theory and Practice
  of Software, {ETAPS} 2021, Luxembourg City, Luxembourg, March 27 - April 1,
  2021, Proceedings}, volume 12650 of \emph{LNCS}, pages 551--571, Springer
  (2021).
\newline\urlprefix\url{https://doi.org/10.1007/978-3-030-71995-1\_28}

\end{thebibliography}


\end{document}